\documentclass[subscriptcorrection,upint,varvw,barcolor=black,mathalfa=cal=euler,balance,hyphenate,french,pdf-a,nolists,nofoot]{asmejour} %

\hypersetup{%
	pdfauthor={Mihails Milehins},
	pdftitle={template.pdf},
	pdfkeywords={Bouc-Wen Model, Collision, Hysteresis, Impact, Rigid Body Dynamics}
}

\usepackage{paralist}
\usepackage{tikz}
\usetikzlibrary{calc}
\usetikzlibrary{arrows.meta}
\usepackage{pgfplots}
\usepackage{physics}
\usepackage{amsthm}

\newtheorem{theorem}{Theorem}[section]
\newtheorem{lem}[theorem]{Lemma}
\newtheorem{prop}[theorem]{Proposition}

\theoremstyle{definition}
\newtheorem{definition}{Definition}[section]
\theoremstyle{remark}

\JourName{Computational and Nonlinear Dynamics}

\begin{document}



\SetAuthorBlock{Mihails Milehins\CorrespondingAuthor}{%
	Department of Mechanical Engineering,\\
	Auburn University,\\
	Auburn, AL 36849,\\
    email: mzm0390@auburn.edu} 

\SetAuthorBlock{Dan B. Marghitu}{
	Department of Mechanical Engineering,\\
	Auburn University,\\
	Auburn, AL 36849,\\
    email: marghdb@auburn.edu}

\title{The Bouc-Wen Model for Binary Direct Collinear Collisions of Convex Viscoplastic Bodies}


\begin{abstract}
We study mathematical models of binary direct collinear collisions of convex viscoplastic bodies based on two incremental collision laws that employ the Bouc-Wen differential model of hysteresis to represent the elastoplastic behavior of the materials of the colliding bodies. These collision laws are the Bouc-Wen-Simon-Hunt-Crossley Collision Law (BWSHCCL) and the Bouc-Wen-Maxwell Collision Law (BWMCL). The BWSHCCL comprises of the Bouc-Wen model amended with a nonlinear Hertzian elastic spring element and connected in parallel to a nonlinear displacement-dependent and velocity-dependent energy dissipation element. The BWMCL comprises of the Bouc-Wen model amended with a nonlinear Hertzian elastic spring element and connected in series to a linear velocity-dependent energy dissipation element. The mathematical models of the collision process are presented in the form of finite-dimensional initial value problems. We show that the models possess favorable analytical properties (e.g., global existence, uniqueness, and boundedness of the solutions) under suitable restrictions on the values of their parameters. Furthermore, based on the results of two model parameter identification studies, we demonstrate that good agreement can be attained between experimental data and numerical approximations of the behavior of the mathematical models across a wide range of initial relative velocities of the colliding bodies while using parameterizations of the models that are independent of the initial relative velocity.
\end{abstract}

\date{\today} 

\maketitle 

\section{Introduction}\label{sec:introduction}

The majority of approaches for modeling of systems of rigid bodies with contacts can be classified as nonsmooth dynamics formulations or continuous formulations (e.g., see \cite{flores_modeling_2011, machado_compliant_2012}). In nonsmooth dynamics formulations, a system of rigid bodies is modeled either as a complementarity problem, a differential variational inequality, or a hybrid system in a manner such that interpenetration of bodies in contact is prevented (e.g., see \cite{panagiotopoulos_inequality_1985, pfeiffer_multibody_2004, stewart_dynamics_2011, goebel_hybrid_2012, brogliato_nonsmooth_2016, sanfelice_hybrid_2021}). In continuous formulations, the surfaces of the bodies are modeled using virtual viscoelastic, or viscoplastic elements (e.g., see \cite{terzopoulos_elastically_1987, platt_constraint_1988, moore_collision_1988}). Under most circumstances, the nonsmooth dynamics formulations require an algebraic constitutive law to achieve closure (e.g., see \cite{pfeiffer_multibody_2004, stewart_dynamics_2011, roithmayr_dynamics_2016, stronge_impact_2018, ruina_introduction_2019}), whereas the continuous formulations require a dynamic model that can describe the evolution of the contact force (e.g., see \cite{machado_compliant_2012, corral_nonlinear_2021}). Such constitutive laws and contact force models shall be collectively referred to as collision laws (e.g., see \cite{chatterjee_rigid_1997, chatterjee_two_1998} and \cite[MacSithigh (1995), as cited in][]{chatterjee_rigid_1997}). The algebraic constitutive laws shall be referred to as algebraic collision laws, and the dynamic models that describe the evolution of the contact force shall be referred to as incremental collision laws.  

The collision laws that are studied in this article are incremental collision laws for impacts that are direct and collinear (e.g., see \cite{gilardi_literature_2002, stronge_impact_2018}).\footnote{If a collision is direct and collinear, there is no lateral motion during contact, and, therefore, the effects of the lateral friction can be ignored. See \cite{stronge_impact_2018} for further information.} Such incremental collision laws can be classified based on the assumptions about the materials of the colliding bodies: elastic, plastic, viscoelastic, or viscoplastic (e.g., see \cite{stronge_impact_2018, morro_mathematical_2023}).\footnote{The terms plastic and elastoplastic will be used interchangeably.} A further classification, partially consistent with the rheological classification of materials (e.g., see \cite{reiner_classification_1945, stronge_impact_2018, morro_mathematical_2023}), shall be applied to some of the common viscoelastic and viscoplastic incremental collision laws. A viscoelastic or a viscoplastic incremental collision law shall be referred to as a Kelvin-Voigt-type collision law if the model of the material of the contact interface consists of a viscous (rate-dependent) energy dissipation element that is connected in parallel with an elastic spring element or a rate-independent hysteresis element \cite{thomson_elasticity_1865, meyer_zur_1874, voigt_ueber_1890, butcher_characterizing_2000, stronge_impact_2018}. A viscoelastic or a viscoplastic incremental collision law shall be referred to as a Maxwell-type collision law if the model of the material of the contact interface consists of a viscous energy dissipation element that is connected in series with an elastic spring element or a rate-independent hysteresis element \cite{maxwell_dynamical_1867, butcher_characterizing_2000, stronge_impact_2018}. Other configurations of rheological elements are possible (e.g., see \cite{reiner_classification_1945, butcher_characterizing_2000, morro_mathematical_2023}), but shall remain unclassified in the context of this article. The reviews in \cite{luding_collisions_1998, gilardi_literature_2002, kruggel-emden_review_2007, seifried_role_2010, machado_compliant_2012, brake_effect_2013, khulief_modeling_2013, thornton_investigation_2013, alves_comparative_2015, ahmad_impact_2016, brogliato_nonsmooth_2016, banerjee_historical_2017, skrinjar_review_2018, corral_nonlinear_2021, rodrigues_da_silva_compendium_2022, wang_research_2022, wang_review_2022, flores_contact-impact_2023, ding_review_2024} provide descriptions of many important incremental collision laws that have been proposed in the literature in the past. The remainder of this section also contains a brief overview of previously proposed incremental collision laws and mathematical models of the behavior of the materials that can serve as a foundation for the development of incremental collision laws. 

The simplest incremental collision law for elastic bodies is based on the assumption that the evolution of the contact force is governed by Hooke's Law: $F = k \delta$ \cite{hooke_potentia_1678}. Here, $k \in \mathbb{R}_{> 0}$ denotes the effective stiffness of the contact interface, $\delta \in \mathbb{R}$ denotes the relative displacement of the bodies measured at the contact point along the common normal direction and $F \in \mathbb{R}$ denotes the contact force.\footnote{Notation is explained in Appendix \ref{sec:notation}.} A natural extension of this collision law assumes that the relationship between the relative displacement and the contact force is described by the Hertzian power law: $F = k \abs{\delta}^{n-1} \delta$ with $n \in \mathbb{R}_{\geq 1}$ \cite{hertz_uber_1881, hertz_uber_1882}.

One of the simplest incremental collision laws for viscoelastic bodies, the Kelvin-Voigt collision law, is based on the assumption that the contact interface behaves as a Kelvin-Voigt material: $F = k \delta + c \dot{\delta}$ (\cite{thomson_elasticity_1865, meyer_zur_1874, voigt_ueber_1890}, see also \cite{goldsmith_impact_1960, dubowsky_dynamic_1971, dubowsky_dynamic_1971-1, hunt_coefficient_1975, butcher_characterizing_2000}). Here, $c \in \mathbb{R}_{\geq 0}$ represents the viscous damping coefficient.

Since the initial relative velocity $\dot{\delta}(0)$ of the colliding objects is usually not zero, the Kelvin-Voigt collision law may result in an unphysical discontinuity in the evolution of the contact force \cite{hunt_coefficient_1975}. A linear incremental collision law for viscoelastic bodies that avoids this issue, the Maxwell collision law, is based on the assumption that the contact interface behaves as a Maxwell material: $F = k \delta_1 = c \dot{\delta}_2$ with $\delta = \delta_1 + \delta_2$. Here, $\delta_1$ is the displacement of an internal elastic spring element, $\delta_2$ is the displacement of an internal viscous energy dissipation element \cite{maxwell_dynamical_1867, johnson_contact_1985, butcher_characterizing_2000, argatov_mathematical_2013}. Further linear incremental collision laws can be constructed from other types of models of materials based on their rheological classification. For example, incremental collision laws based on the standard linear solid model \cite{poynting_text-book_1902, jeffreys_viscosity_1917, ishlinsky_vibrations_1940, zener_elasticity_1948, morro_mathematical_2023} were studied in \cite{mills_robotic_1992, butcher_characterizing_2000, argatov_mathematical_2013}.

A nonlinear incremental collision law $F = k \abs{\delta}^{n - 1} \delta + \chi \abs{\delta}^n \dot{\delta}$ with $n \in \mathbb{R}_{\geq 1}$ that combines the Hertzian elastic spring model with a nonlinear displacement-dependent energy dissipation term was proposed independently in \cite[Simon (1967), as cited in][]{brogliato_nonsmooth_2016} under a restricted range of parameters and in \cite{hunt_coefficient_1975} (see also \cite{veluswami_multiple_1975, veluswami_multiple_1975-1}) in its full generality. Here, $\chi \in \mathbb{R}_{\geq 0}$ denotes the viscous damping coefficient. Due to the displacement-dependent energy dissipation term, the collision law avoids the discontinuity in the evolution of the contact force that is associated with the Kelvin-Voigt collision law. Usually, it is assumed that $n$ and $k$ are constant for a given choice of geometry and materials of the colliding bodies, whereas $\chi$ may also depend on the initial relative velocity of the colliding bodies. The choice of the parameters of the model is a subject of ongoing research (e.g., see \cite{herbert_shape_1977, lee_dynamics_1983, lankarani_contact_1990, lankarani_continuous_1994, stoianovici_critical_1996, marhefka_compliant_1999, gonthier_regularized_2004, zhang_compliant_2004, zhiying_analysis_2006, ye_note_2009, hu_determination_2011, flores_continuous_2011, gharib_new_2012, khatiwada_generic_2014, jacobs_modeling_2015, hu_dissipative_2015, wang_modeling_2018, carvalho_exact_2019, sherif_models_2019, safaeifar_new_2020, yu_improved_2020, zhang_continuous_2020, zhao_spring-damping_2021, wang_energy_2022, ramaswamy_continuous_2023, tan_influence_2023, sheikhi_azqandi_optimal_2024}). 

Since the publication of \cite{hunt_coefficient_1975}, several authors proposed or employed a number of alternative nonlinear incremental collision laws for viscoelastic bodies of the Kelvin-Voigt type \cite{tatara_study_1982, kuwabara_restitution_1987, ristow_simulating_1992, tsuji_lagrangian_1992, lee_angle_1993, luding_anomalous_1994, hertzsch_low-velocity_1995, shafer_force_1996, brilliantov_collision_1996, brilliantov_model_1996, morgado_energy_1997, falcon_behavior_1998, schwager_coefficient_1998, jankowski_non-linear_2005, jankowski_analytical_2006, bordbar_modeling_2007, schwager_coefficient_2008, choi_efficient_2010, mahmoud_modified_2011, muller_collision_2011, roy_damping_2012, zheng_finite_2012, alizadeh_development_2013, azad_new_2014, brilliantov_dissipative_2015, goldobin_collision_2015, wang_advanced_2017, kaviani_rad_frictional_2018, wang_contact-impact_2019, poursina_optimal_2020, poursina_characterization_2020, wang_further_2020, zhang_continuous_2020, chatterjee_approximate_2022, jia_application_2022, zhang_continuous_2022, zhang_continuous_2022-1, nikravesh_determination_2023, wang_development_2023-1, huang_general_2024, poursina_new_2024, wang_enhanced_2024, zhang_continuous_2024, wang_characteristics_2025}. Most of these models have the form $F = k \abs{\delta}^{n - 1} \delta + \chi \abs{\delta}^p \dot{\delta}$ with $p \in \mathbb{R}_{\geq 0}$ such that $p \neq n$ and/or parameters that differ during the compression and the restitution phases of the collision process. Further incremental collision laws for viscoelastic bodies that are not of the Kelvin-Voigt type were proposed or employed in \cite{khusid_collision_1986, atanackovic_viscoelastic_2004, argatov_mathematical_2013, zbiciak_dynamic_2015, argatov_impact_2016, jian_normal_2019, askari_mathematical_2021, argatov_viscoelastic_2024, ding_approximate_2024}. 

Under some of the common conditions, the materials of the colliding bodies undergo plastic deformation (e.g., see \cite{tresca_memoire_1869, saint-venant_memoire_1871, crook_study_1952, barnhart_transverse_1955, barnhart_stresses_1957, goldsmith_impact_1960}). In this case, there exists a permanent non-zero relative displacement at the time of the separation of the colliding bodies. Furthermore, it is often the case that the energy dissipation during contact is largely a rate-independent phenomenon. While some of the incremental collision laws that were designed for the collisions of viscoelastic bodies also yield a non-zero relative displacement at the time of the separation (e.g., see \cite{poursina_new_2024}), the energy dissipation mechanisms associated with these collision laws are strongly dependent on the relative velocity of the bodies in contact. Therefore, arguably, these collision laws are less suitable for the description of the collisions of plastic and viscoplastic bodies.

The simplest incremental collision laws for plastic bodies are rate-independent and based on the assumption that the relationship between the relative displacement and the contact force is different in compression and restitution (\cite{tresca_memoire_1869, saint-venant_memoire_1871, crook_study_1952, barnhart_transverse_1955, barnhart_stresses_1957, goldsmith_impact_1960, walton_viscosity_1986, lankarani_continuous_1994, gharib_new_2012, iqbal_coefficient_2025} and \cite[Kadomtsev (1990), as cited in][]{biryukov_dynamic_2002}). In \cite{andrews_theory_1930, tabor_simple_1948, tabor_hardness_1951, johnson_contact_1985, chang_normal_1992, ning_elastic-plastic_1993, sadd_contact_1993, yigit_impact_1994, yigit_use_1995, thornton_coefficient_1997, vu-quoc_elastoplastic_1999, zhao_asperity_1999, beig_contact_2000, li_theoretical_2001, etsion_unloading_2005, weir_coefficient_2005, mangwandi_coefficient_2007, luding_cohesive_2008, du_energy_2009, antonyuk_energy_2010, jackson_predicting_2010, brake_analytical_2012, brake_analytical_2015, ma_contact_2015, ghaednia_comprehensive_2016, mukhopadhyay_theoretical_2023} and \cite[Kil'chevskii (1976), Aleksandrov et al (1984), Aleksandrov and Romalis (1986), as cited in][]{biryukov_dynamic_2002}, the authors proposed or employed several multi-stage (e.g., elastic loading, elastic-plastic loading, elastic unloading) incremental collision laws for plastic bodies. A variety of incremental collision laws for viscoplastic bodies were proposed or employed in \cite{storakers_elastic_2000, tomas_particle_2000, ismail_impact_2008, yigit_nonlinear_2011, burgoyne_strain-rate-dependent_2014, christoforou_inelastic_2016, ahmad_improved_2016, borovin_nonlinear_2019, wang_extension_2022, wang_development_2023, lin_blade-coating_2024, wang_enhanced_2024}. The majority of these collision laws combine a bilinear or multi-stage elastoplastic element with a linear viscous energy dissipation element.

A different line of research in the area of incremental collision laws was initiated in \cite{kikuuwe_incorporating_2007} and continued in \cite{xiong_differential_2013, xiong_contact_2014}. The collision laws proposed in \cite{kikuuwe_incorporating_2007, xiong_differential_2013, xiong_contact_2014} were designed to be suitable for the description of a wide variety of collision phenomena. Their development was also driven by an attempt to overcome various disadvantages of some of the traditional incremental collision laws. These disadvantages include an unnatural sticking force that can appear when using the Simon-Hunt-Crossley collision law (and some of its extensions) under certain conditions, the inability of the Simon-Hunt-Crossley collision law (and some of its extensions) to describe the non-zero indentation at the time of the separation (usually due to either plastic deformation \cite{tresca_memoire_1869, saint-venant_memoire_1871, crook_study_1952} or elastic aftereffect \cite{boltzmann_theorie_1874, wang_advanced_2017}), as well as the piecewise nature of the traditional collisions laws for viscoplastic bodies. While the aforementioned collision laws overcome these issues, they were originally expressed in the form of differential-algebraic inclusions. Nonetheless, the collision laws proposed in \cite{xiong_differential_2013, xiong_contact_2014} were also reformulated as ordinary differential equations.

As mentioned in \cite{butcher_characterizing_2000, brogliato_nonsmooth_2016}, almost any model that can describe the dynamic behavior of a material can be used as a foundation for the development of incremental collision laws. There exist several general models of hysteresis that are suitable for the description of the behavior of plastic or viscoplastic materials (e.g., see \cite{saint-venant_memoire_1871, schwedoff_recherches_1889, bingham_fluidity_1922, masing_eigenspannungen_1926, schofield_relationship_1932, ramberg_description_1943, reiner_classification_1945, pisarenko_vibrations_1962, jennings_response_1963, rosenblueth_kind_1964, iwan_distributed-element_1966, bouc_modemathematique_1971, ozdemir_nonlinear_1976, chen_constitutive_1980, jayakumar_modeling_1987, monteiro_marques_existence_1994, bastien_study_2000, charalampakis_boucwen_2009, biswas_reduced-order_2014, biswas_two-state_2015, brogliato_nonsmooth_2016, morro_mathematical_2023}). However, little research has been done to investigate their applicability to the construction of general-purpose incremental collision laws for viscoplastic bodies, although limited progress has been made (e.g., \cite{antonyuk_energy_2010, xiong_differential_2013, xiong_contact_2014, brogliato_nonsmooth_2016, borovin_nonlinear_2019, maleki_modified_2023, maleki_modified_2024} may be considered relevant in this context).

This article showcases a study of several incremental collision laws based on the Bouc-Wen differential model of hysteresis. The Bouc-Wen model is a general parameterizable rate-independent differential model of hysteresis. It was proposed in \cite{bouc_forced_1967, bouc_modemathematique_1971} and extended in \cite{wen_method_1976}. The model and its extensions (e.g., see \cite{sivaselvan_hysteretic_2000} and \cite{charalampakis_boucwen_2009}) have been used successfully in a variety of fields (e.g., see \cite{ikhouane_systems_2007}), including vibro-impacts \cite{maleki_modified_2023, maleki_modified_2024}. In the context of the study described in this article, the Bouc-Wen model was chosen due to its simplicity, popularity, and its wide scope of applicability.

\section{Contributions and Outline}\label{sec:contributions}

The primary contributions of this article are a description of an analytical study of two mathematical models of binary direct collinear collisions of convex viscoplastic bodies based on the Bouc-Wen model of hysteresis, and a description of two model parameter identification studies showcasing the possibility of attainment of good agreement between the experimental data and the data stemming from the numerical simulations of the aforementioned mathematical models across a wide range of initial relative velocities of the colliding bodies while using parameterizations of the models that are independent of the initial relative velocity. The article also provides a review of recent research on the subject of incremental collision laws for binary frictionless collisions.

The remainder of the article is organized as follows:
\begin{compactitem}
\item Section \ref{sec:mps} introduces a high-level model of the physical system that is studied in the remainder of the article.
\item Section \ref{sec:BWM} introduces the Bouc-Wen model and a trivial incremental collision law based on the Bouc-Wen model, the Bouc-Wen-Hertz Collision Law.
\item Section \ref{sec:KVCL} introduces the Bouc-Wen-Simon-Hunt-Crossley Collision Law, a Kelvin-Voigt-type collision law based on the Bouc-Wen model.
\item Section \ref{sec:MCL} introduces the Bouc-Wen-Maxwell Collision Law, a Maxwell-type collision law based on the Bouc-Wen model.
\item Section \ref{sec:ID} describes two methodologies for the identification of the parameters for the collision laws.
\item Section \ref{sec:conclusions} provides conclusions and recommendations.
\item Appendices \ref{sec:notation}-\ref{sec:analysis_BWMCM} describe the mathematical conventions and provide the proofs of the main results.
\end{compactitem}

\section{Model of the Physical System}\label{sec:mps}

The discussion that follows is with reference to Fig. \ref{fig:nc}. As mentioned previously, mathematical notation is explained in Appendix \ref{sec:notation}. The notational conventions for mechanics are adopted from \cite{roithmayr_dynamics_2016} and \cite{stronge_impact_2018}.  The units are seldom stated explicitly: it is assumed that a consistent system of units is used for all dimensional quantities. The boldfaced symbols will denote vectors, with the notation $\mathbf{r}_{B/A}$ reserved for the displacement of the point $B$ relative to the point $A$. Let $\mathcal{N}$ denote an inertial frame of reference with the inertial origin $O$, and let $\hat{\mathbf{n}} = (\hat{\mathbf{n}}_1, \hat{\mathbf{n}}_2, \hat{\mathbf{n}}_3)$ be the canonical right-handed orthonormal coordinate system. The circumflex over baldfaced letters shall be used to indicate normalized vectors. Any vector $\mathbf{v}$ can be identified unambiguously with the set of its components in $\hat{\mathbf{n}}$: $\mathbf{v} \triangleq v_1 \hat{\mathbf{n}}_1 + v_2 \hat{\mathbf{n}}_2 + v_3 \hat{\mathbf{n}}_3$. 

Suppose that $\mathcal{B}_1$ is a compact and strictly convex rigid body. Suppose that $\mathcal{B}_2$ is a convex rigid body with a topologically smooth surface. The bodies are assumed to come into contact at the point $C \triangleq O$\footnote{$C$ is merely an abbreviation for $O$. Both symbols represent the same point in space.} at the time $t_0 \triangleq 0 \in \mathbb{R}$, with the tangent plane spanned by $\hat{\mathbf{n}}_2$ and $\hat{\mathbf{n}}_3$, and with the common normal direction $\hat{\mathbf{n}}_1$. Suppose that the center of mass of $\mathcal{B}_i$\footnote{The index $i$ ranges over $\{ 1, 2\}$ here and in the remainder of this section.} is located at $G_i$, which lies on the line $AB$ through $C$ and parallel to $\hat{\mathbf{n}}_1$. The point located on the boundary of $\mathcal{B}_i$ that coincides with $C$ will be denoted $C_i$. For notational convenience, define $\mathbf{r}_{i} \triangleq \mathbf{r}_{C_i/G_i}$. Without the loss of generality, it shall be assumed that $\mathbf{r}_1 \cdot \hat{\mathbf{n}}_1 < 0$ and $\mathbf{r}_2 \cdot \hat{\mathbf{n}}_1 > 0$. The velocity fields of both bodies are assumed to be uniform and parallel to the line $AB$. The configuration, as hereinbefore described, corresponds to a binary direct collinear impact (e.g., see \cite{stronge_impact_2018}).

Referring to \cite{stronge_impact_2018}, it shall be assumed that while the bodies remain in contact, the motion of the system is governed by the laws of rigid body dynamics (Newton, \cite{newton_mathematical_1729}), with the contact point described as an infinitesimal deformable particle. The mass of $\mathcal{B}_i$ shall be denoted as $m_i \in \mathbb{R}_{>0}$. The force that acts on the body $\mathcal{B}_i$ at the contact point $C_i$ shall be denoted as $\mathbf{F}^{C_i}$. It is postulated that 
\begin{equation}
\mathbf{F}^{C_1} = -\mathbf{F}^{C_2} \triangleq \mathbf{F} = F \hat{\mathbf{n}}_1
\end{equation} 
for $F(\cdot) \in \mathbb{R}$ with $F(0) = 0$. Due to the nature of the collision process, a single generalized coordinate is sufficient to describe the motion of each body during contact: $x_i$ shall refer to the displacement of $G_i$ along $\mathbf{n}_1$ from its initial position. The equations of motion are $\ddot{x}_1 = m_1^{-1} F$ and $\ddot{x}_2 = - m_2^{-1} F$, with $x_1(0) = x_2(0) = 0$, $\dot{x}_1(0) = v_{1,0} \in \mathbb{R}$, and $\dot{x}_2(0) = v_{2,0} \in \mathbb{R}$ such that $v_0 \triangleq -(v_{1,0} - v_{2, 0}) \in \mathbb{R}_{>0}$. Denoting $m \triangleq m_1 m_2 (m_1 + m_2)^{-1}$, $x \triangleq x_1 - x_2$, and $v \triangleq \dot{x} = \dot{x}_1 - \dot{x}_2$, the equations of motion can be transformed to
\begin{equation}\label{eq:main}
\begin{cases}
\dot{x} = v & x(0) = 0\\
\dot{v} = m^{-1} F & v(0) = -v_0 
\end{cases}
\end{equation}
The form of the force $F$ depends on the chosen collision law. As an aside, it should be noted that if, by abuse of notation, $m_2 = +\infty$, then $m^{-1} = m_1^{-1}$, which corresponds to the collision of a body $\mathcal{B}_1$ of finite mass with a stationary body $\mathcal{B}_2$.

Provided that a solution of the initial value problem (IVP) given by Eq. \eqref{eq:main} (including any possible amendments associated with $F$) exists and is unique on a non-degenerate time interval $I \subseteq \mathbb{R}$ with $0 \in I$, the time of the separation $t_s \in \mathbb{R}_{> 0} \cup \{ + \infty \}$ shall be defined as
\begin{equation}\label{eq:main_t_s}
t_s \triangleq \inf \{ t \in I_{\geq 0} : F(t) \leq 0 \wedge 0 \leq v(t) \}
\end{equation}
Since $v$ is continuous and $v(0) < 0$, $t_s$ is well defined. It is important to note that $t_s$ may not be finite.

Informally, the Coefficient of Restitution (CoR) $e \in \mathbb{R}$ for binary direct collinear collisions can be defined as the additive inverse of the value of the ratio of the relative velocity at the time of the separation $v(t_s)$ to the value of the relative velocity at the time of the collision $v(0)$ (e.g., see \cite{stronge_impact_2018}).\footnote{It should be remarked that the type of CoR that is employed in this study is usually referred to as the kinematic CoR and attributed to Sir Isaac Newton \cite{newton_mathematical_1729}. However, there exist other types of CoRs, such as the kinetic CoR due to Siméon Denis Poisson \cite{poisson_treatise_1842}, and the energetic CoR due to William Stronge (e.g., see \cite{stronge_rigid_1990} and \cite{stronge_impact_2018}). It should also be remarked that no explicit restrictions are imposed on the value of the CoR in this article, but normally it lies in the interval $[0, 1] \subseteq \mathbb{R}$.} More formally, for any physical system described by Eq. \eqref{eq:main}, $e$ shall be given by 
\begin{equation}\label{eq:main_cor}
e \triangleq 
\begin{cases}
-v(t_s)/v(0) & t_s \neq +\infty \\
0 & t_s = +\infty
\end{cases}
\end{equation}
provided that the solution of the IVP given by Eq. \eqref{eq:main} exists and is unique on some non-degenerate time interval $I \subseteq \mathbb{R}$ with $0 \in I$.

\begin{figure}
\centering
\begin{tikzpicture}[>={Stealth[scale=0.6]}]

	\coordinate (O) at (0, 0);
	\coordinate (G1) at (0, 1);
	\coordinate (G2) at (0, -2);
	\coordinate (E1) at (0, 0.5);
	\coordinate (E2) at (0, -0.5);

    \draw[fill=black](0,0) circle (1 pt) node [above  right] {$C$};
    \draw[fill=black](G1) circle (1 pt) node [left] {$G_1$};
    \draw[fill=black](G2) circle (1 pt) node [left] {$G_2$};

    \draw[->, line width=1](3,1) -- (3,2) node [right] {$\hat{\mathbf{n}}_1$};
    \draw[->, line width=1](3,1) -- (2,1) node [above] {$\hat{\mathbf{n}}_2$};

    \draw[dotted, line width=1](-2.5,0) -- (2.5,0);
    \draw[dotted, line width=1](0,2) node [right] {$A$} -- (0,-3.5) node [right] {$B$};    
    
    \draw[->, line width=1](O) -- (0,0.5) node [left] {$\mathbf{F}$};
    \draw[->, line width=1](O) -- (0,-0.5) node [left] {$\mathbf{F}$};
    
    \draw[->, line width=0.75, dashed](-2,1.25) node [left] {$\mathcal{B}_1$} -- (-1.25,1);
    \draw[->, line width=0.75, dashed](-1.5,-1.5) node [left] {$\mathcal{B}_2$} -- (-0.5,-1.5);
    
    \draw plot [smooth cycle, tension=.8] coordinates {(0,0) (1.5,1) (0,1.5) (-1.5,1)};
    \draw plot [smooth cycle, tension=.8] coordinates {(0,0) (0.75,-2) (0,-3) (-0.75,-2)};
    
    \draw (0,0.2) -- (-0.2,0.2) -- (-0.2,0);
    \draw (3,1.2) -- (2.8,1.2) -- (2.8,1);

\end{tikzpicture}
\caption{System diagram}\label{fig:nc}
\end{figure}

\section{The Bouc-Wen-Hertz Collision Law}\label{sec:BWM}

The primary references for the Bouc-Wen differential model of hysteresis are \cite{ikhouane_systems_2007, ikhouane_bounded_2004, ikhouane_hysteretic_2005, ikhouane_hysteretic_2005-1, ikhouane_analytical_2006, ikhouane_dynamic_2007, ikhouane_variation_2007}. A collision law based on the Bouc-Wen model, the Bouc-Wen Collision Law (BWCL), can be specified as\footnote{This form of the Bouc-Wen model was employed, for example, in \cite{ma_parameter_2004}.}
\begin{equation}\label{eq:BWCL}
\begin{cases}
\dot{x} = u\\
\dot{z} = A u - \beta \abs{z}^{n - 1} z \abs{u} - \gamma \abs{z}^n u \\
F =  - \alpha k x - \alpha_c k z
\end{cases}
\end{equation}
Here, $x, z \in \mathbb{R}$ are internal state variables, $u \in \mathbb{R}$ is an input variable that is meant to represent the relative velocity of the colliding bodies (e.g., $v$ in Eq. \eqref{eq:main}), $F \in \mathbb{R}$ is an output variable that is meant to represent the contact force between the colliding bodies. The model is parameterized by $A, k \in \mathbb{R}_{>0}$, $\alpha \in (0, 1)$, $\beta \in \mathbb{R}_{\geq 0}$, $\gamma \in [- \beta, \beta]$, and $n \in \mathbb{R}_{\geq 1}$, with $\alpha_c \triangleq 1 - \alpha$.\footnote{In what follows, $\alpha_c$ will always be used as an abbreviation for $1 - \alpha$.}

To accommodate the nonlinearities that are present in some of the traditional contact force models \cite{hertz_uber_1881, hertz_uber_1882}, the output function of the BWCL is augmented to yield\footnote{The form of the output $F$ was chosen heuristically. The choice of the best form of $F$ that takes into account the Hertzian nonlinearity is a potential avenue for future research.}
\begin{equation}\label{eq:BWHCL}
\begin{cases}
\dot{x} = u\\
\dot{z} = A u - \beta \abs{z}^{n - 1} z \abs{u} - \gamma \abs{z}^n u \\
F = - \alpha k \abs{x}^{p - 1} x - \alpha_c k \abs{z}^{p - 1} z
\end{cases}
\end{equation}
Here, $p \in \mathbb{R}_{\geq 1}$ is an additional parameter. The model given by Eq. \eqref{eq:BWHCL} shall be referred to as the Bouc-Wen-Hertz Collision Law (BWHCL).

\section{The Bouc-Wen-Simon-Hunt-Crossley Collision Law}\label{sec:KVCL}

The BWHCL can be augmented further to yield a Kelvin-Voigt-type collision law:
\begin{equation}\label{eq:BWSHCCL}
\begin{cases}
\dot{x} = u\\
\dot{z} = A u - \beta \abs{z}^{n - 1} z \abs{u} - \gamma \abs{z}^n u \\
F = - \alpha k \abs{x}^{p - 1} x - \alpha_c k \abs{z}^{p - 1} z - c \abs{x}^p u
\end{cases}
\end{equation}
Here, $c \in \mathbb{R}_{\geq 0}$ is an additional parameter. This collision law was inspired by the Simon-Hunt-Crossley collision law and shall be referred to as the Bouc-Wen-Simon-Hunt-Crossley Collision Law (BWSHCCL). A diagrammatic representation of the lumped element model upon which the BWSHCCL is based is shown in Fig. \ref{fig:BWSHCCL_LEM}.

\begin{figure}
\centering
\begin{tikzpicture}[>={Stealth[scale=0.6]}]

    \draw[line width=1](0, 0) -- (1, 0);
    \draw[line width=1](1, -0.25) -- (1, 0.25);
    \draw[line width=1](0.75, 0.25) -- (1.25, 0.25);
    \draw[line width=1](0.75, -0.25) -- (1.25, -0.25);
    \draw[line width=1](1.25, -0.25) -- (1.25, 0.25);
    \draw[line width=1](1.25, 0) -- (2, 0);
    \draw plot [smooth, tension=.8] coordinates {(0.5, -0.5) (0.7, -0.25) (1, 0) (1.3, 0.25) (1.5, 0.5)};    

    \draw[line width=1](0, -2) -- (0.25, -2);  
	\draw[draw=black] (0.25, -2.5) rectangle node{BWHCL} (1.75, -1.5);		
    \draw[line width=1](1.75, -2) -- (2, -2); 
    
    \draw[fill=black](0, -1) circle (2 pt);
    \draw[fill=black](2, -1) circle (2 pt);
    
    \draw[line width=1](0, 0) -- (0, -2);
  	\draw[line width=1](2, 0) -- (2, -2);
  	
    \draw[line width=1, dashed](0, -2) -- (0, -3.25);
    \draw[line width=1, dashed](2, -2) -- (2, -3.25);

    \draw[<->, line width=1](0, -3) -- node[above] {$x$} (2, -3);

\end{tikzpicture}
\caption{BWSHCCL: diagrammatic representation of the lumped element model (the upper part of the diagram depicts the nonlinear viscous energy dissipation element, the lower part of the diagram depicts the BWHCL element)}\label{fig:BWSHCCL_LEM}
\end{figure}

A feedback interconnection of the abstract collision model given by Eq. \eqref{eq:main} and the BWSHCCL results in the following model:
\begin{equation}\label{eq:BWSHCCM}
\begin{cases}
\dot{x} = v \\
\dot{z} = A v - \beta \abs{z}^{n - 1} z \abs{v} - \gamma \abs{z}^n v \\
\dot{v} = - \alpha \frac{k}{m} \abs{x}^{p - 1} x - \alpha_c \frac{k}{m} \abs{z}^{p - 1} z - \frac{c}{m} \abs{x}^p v \\
\begin{matrix} x(0) = 0, & z(0) = 0, & v(0) = -v_0 \end{matrix}
\end{cases}
\end{equation}
The model given by Eq. \eqref{eq:BWSHCCM} shall be referred to as the Bouc-Wen-Simon-Hunt-Crossley Collision Model (BWSHCCM).

The BWSHCCM will now be nondimensionalized.\footnote{See \cite{logan_applied_2013} for a description of the methodology that was employed for the nondimensionalization of the BWSHCCM.} 
The relevant fundamental dimensions are mass $M$, length $L$, and time $T$. The dimensions of the state variables are given by $[x] = L$, $[z] = L$, and $[v] = L T^{-1}$. The parameters $A$, $\alpha$, $n$, and $p$ are dimensionless. The dimensions of the remaining parameters are given by $[m] = M$, $[k] = M L^{1 - p} T^{-2}$, $[c] = M L^{-p} T^{-1}$, $[\beta] = L^{-n}$, $[\gamma] = L^{-n}$, and $[v_0] = L T^{-1}$. The nondimensionalized model will evolve with respect to the nondimensionalized time variable $T \triangleq t/T_c$, with the time scale $T_c \in \mathbb{R}_{> 0}$ given by
\begin{equation}\label{eq:BWSHCCM_T_c}
T_c \triangleq \left( \frac{1}{\alpha + \alpha_c A^p} \right)^{\frac{1}{p + 1}} \left( \frac{m}{k} \right)^{\frac{1}{p + 1}} v_0^{-\frac{p - 1}{p + 1}}
\end{equation}
The nondimensionalized state variables are given by $X \triangleq x/X_c$, $Z \triangleq z/Z_c$, $V \triangleq v/(X_c/T_c)$, with the spatial scales $X_c, Z_c \in \mathbb{R}_{>0}$ given by
\begin{equation}\label{eq:BWSHCCM_X_c}
X_c \triangleq \left( \frac{1}{\alpha + \alpha_c A^p} \right)^{\frac{1}{p + 1}} \left( \frac{m}{k} \right)^{\frac{1}{p + 1}} v_0^{\frac{2}{p + 1}}
\end{equation}
\begin{equation}\label{eq:BWSHCCM_Z_c}
Z_c \triangleq \left( \frac{1}{\alpha + \alpha_c A^p} \right)^{\frac{1}{p + 1}} A \left( \frac{m}{k} \right)^{\frac{1}{p + 1}} v_0^{\frac{2}{p + 1}}
\end{equation}
respectively. Introduction of the dimensionless parameters 
\begin{equation}
B \triangleq \left( \frac{A^{p + 1}}{\alpha + \alpha_c A^p} \right)^{\frac{n}{p + 1}} \frac{\beta}{A} \left( \frac{m}{k} \right)^{\frac{n}{p + 1}} v_0^{\frac{2 n}{p + 1}}
\end{equation}
\begin{equation}
\mathit{\Gamma} \triangleq \left( \frac{A^{p + 1}}{\alpha + \alpha_c A^p} \right)^{\frac{n}{p + 1}} \frac{\gamma}{A} \left( \frac{m}{k} \right)^{\frac{n}{p + 1}} v_0^{\frac{2 n}{p + 1}}
\end{equation}
\begin{equation}
\kappa \triangleq \frac{\alpha}{\alpha + \alpha_c A^p}
\end{equation}
\begin{equation}
\sigma \triangleq \frac{1}{\alpha + \alpha_c A^p} \frac{c}{k} v_0
\end{equation}
and the abbreviation $\kappa_c \triangleq 1 - \kappa$, and nondimensionalization of Eq. \eqref{eq:BWSHCCM} results in the following model:
\begin{equation}\label{eq:BWSHCCM_nd}
\begin{cases}
\dot{X} = V \\
\dot{Z} = V - B \abs{Z}^{n - 1} Z \abs{V} - \mathit{\Gamma} \abs{Z}^n V \\
\dot{V} = - \kappa \abs{X}^{p - 1} X - \kappa_c \abs{Z}^{p - 1} Z - \sigma \abs{X}^p V \\
\begin{matrix} X(0) = 0, & Z(0) = 0, & V(0) = -1 \end{matrix}
\end{cases}
\end{equation}
This model shall be referred to as the Nondimensionalized Bouc-Wen-Simon-Hunt-Crossley Collision Model (NDBWSHCCM). Most of the further analysis will be based on the NDBWSHCCM rather than the BWSHCCM. 

Under the assumption that the NDBWSHCCM is parameterized by $B \in \mathbb{R}_{> 0}$, $\mathit{\Gamma} \in (-B, B)$, $\kappa \in (0, 1)$, $\sigma \in \mathbb{R}_{>0}$, $n, p \in \mathbb{R}_{\geq 1}$, there exists a unique bounded solution of the NDBWSHCCM on any time interval $[0, T_e)$ with $T_e \in \mathbb{R}_{>0} \cup \{ +\infty \}$. The set of equilibrium points of the NDBWSHCCM is
\begin{equation}
\mathcal{E} \triangleq \left\{ \left(X, - \left(\frac{\kappa}{\kappa_c} \right)^{\frac{1}{p}} X, 0 \right) : X \in \mathbb{R} \right\}
\end{equation}
Each solution of the NDBWSHCCM converges to a subset of $\mathcal{E}$ at a finite distance from the origin. See Appendix \ref{sec:analysis_BWSHCCM} for further details.\footnote{There is no reason to believe that the chosen range of the parameters provides a necessary condition for any of the results stated in this paragraph. However, arguably, the range is sufficiently wide for most engineering applications.} 

In applications, it is often of interest to understand how a given physical system behaves with respect to the changes in $v_0$, the absolute value of the relative velocity immediately before the collision. While $v_0$ appears only in the initial condition in the BWSHCCM, multiple parameters of the NDBWSHCCM depend on $v_0$. The dependence of the parameters of the NDBWSHCCM on $v_0$ can be described explicitly by the function $\mathcal{P} : \mathbb{P}^{*} \times \mathbb{R}_{>0} \longrightarrow \mathbb{P}$ that maps $(B_b, \mathit{\Gamma}_b, \kappa, \sigma_b, n, p) \in \mathbb{P}^{*}$ and $v_0 \in \mathbb{R}_{> 0}$ to 
\[
\left( B_b v_0^{\frac{2 n}{p + 1}}, \mathit{\Gamma}_b v_0^{\frac{2 n}{p + 1}}, \kappa, \sigma_b v_0, n, p \right) \in \mathbb{P}
\] 
where $\mathbb{P}^{*} = \mathbb{P} \subseteq \mathbb{R}^6$ consist of all $P = (B, \mathit{\Gamma}, \kappa, \sigma, n, p)$ such that $B \in \mathbb{R}_{> 0}$, $\mathit{\Gamma} \in (-B, B)$, $\kappa \in (0, 1)$, $\sigma \in \mathbb{R}_{>0}$, $n, p \in \mathbb{R}_{\geq 1}$. It should be noted that different symbols are used for $\mathbb{P}^{*}$ and $\mathbb{P}$ because (informally) they carry different semantics and they are meant to be used in different contexts. However, the sets are identical from the perspective of set theory. The elements of $\mathbb{P}^{*}$ shall be referred to as the base parameters of the NDBWSHCCM. The base parameters are nothing more than convenient abstractions for the study of the behavior of a given physical system represented by the NDBWSHCCM with respect to the changes in the initial relative velocity (e.g., see Section \ref{sec:ID}).

The function $\Phi : \mathbb{P} \times \mathbb{R}_{\geq 0} \longrightarrow \mathbb{R}^3$ shall be defined in a manner such that $\Phi_P(T)$ represents the value of the solution of the NDBWSHCCM parameterized by $P \in \mathbb{P}$ at the time $T \in \mathbb{R}_{\geq 0}$. The contact force $F : \mathbb{P} \times \mathbb{R}^3 \longrightarrow \mathbb{R}$ for the NDBWSHCCM shall be defined as 
\begin{equation}
F_P(X, Z, V) \triangleq - \kappa \abs{X}^{p - 1} X - \kappa_c \abs{Z}^{p - 1} Z - \sigma \abs{X}^p V
\end{equation}
for any $(X, Z, V) \in \mathbb{R}^3$ and $P \in \mathbb{P}$ such that $\kappa = P_3$, $\sigma = P_4$, and $p = P_6$. With reference to Eq. \eqref{eq:main_t_s}, the time of the separation $T_s : \mathbb{P} \longrightarrow \mathbb{R}_{> 0} \cup \{ +\infty \}$ for the NDBWSHCCM shall be defined as
\begin{equation}
T_s (P) \triangleq \inf \{ T \in \mathbb{R}_{\geq 0} : F_P(\Phi_P(T)) \leq 0 \wedge 0 \leq \Phi_{P,3}(T) \}
\end{equation}
for all $P \in \mathbb{P}$. With reference to Eq. \eqref{eq:main_cor}, CoR $e : \mathbb{P} \longrightarrow \mathbb{R}$ for the NDBWSHCCM shall be defined as
\begin{equation}
e(P) \triangleq 
\begin{cases}
\Phi_{P,3}(T_s(P)) & T_s(P) \neq +\infty \\
0 & T_s(P) = +\infty
\end{cases}
\end{equation}
for all $P \in \mathbb{P}$.

\section{The Bouc-Wen-Maxwell Collision Law}\label{sec:MCL}

The BWHCL can also be augmented further to yield a Maxwell-type collision law:
\begin{equation}\label{eq:BWMCL}
\begin{cases}
\dot{r} = \alpha \frac{k}{c} \abs{y}^{p - 1} y + \alpha_c \frac{k}{c} \abs{z}^{p - 1} z\\
\dot{y} = - \dot{r} + u\\
\dot{z} = A \dot{y} - \beta \abs{z}^{n - 1} z \abs{\dot{y}} - \gamma \abs{z}^n \dot{y} \\
F = - c \dot{r} = - \alpha k \abs{y}^{p - 1} y - \alpha_c k \abs{z}^{p - 1} z
\end{cases}
\end{equation}
Here, $r, y, z \in \mathbb{R}$ are internal state variables, $u \in \mathbb{R}$ is an input variable that is meant to represent the relative velocity of the colliding bodies, $F \in \mathbb{R}$ is an output variable that is meant to represent the contact force between the colliding bodies, and $c \in \mathbb{R}_{> 0}$ is an additional parameter. This collision law shall be referred to as the Bouc-Wen-Maxwell Collision Law (BWMCL). A diagrammatic representation of the lumped element model upon which the BWMCL is based is shown in Fig. \ref{fig:BWMCL_LEM}.

\begin{figure}
\centering
\begin{tikzpicture}[>={Stealth[scale=0.6]}]

    \draw[fill=black](0, 0) circle (2 pt);

    \draw[line width=1](0, 0) -- (1, 0);
    \draw[line width=1](1, -0.25) -- (1, 0.25);
    \draw[line width=1](0.75, 0.25) -- (1.25, 0.25);
    \draw[line width=1](0.75, -0.25) -- (1.25, -0.25);
    \draw[line width=1](1.25, -0.25) -- (1.25, 0.25);
    \draw[line width=1](1.25, 0) -- (2, 0);
    
    \draw[fill=black](2, 0) circle (2 pt);
    
    \draw[line width=1](2, 0) -- (2.25, 0);  
	\draw[draw=black, line width=1] (2.25, -0.5) rectangle node{BWHCL} (3.75, 0.5);		
    \draw[line width=1](3.75, 0) -- (4, 0); 
    
    \draw[fill=black](4, 0) circle (2 pt);
    
    \draw[line width=1, dashed](0, 0) -- (0, -1.75);
    \draw[line width=1, dashed](2, 0) -- (2, -1.25);
    \draw[line width=1, dashed](4, 0) -- (4, -1.75);
    
    \draw[<->, line width=1](0, -1) -- node[above] {$r$} (2, -1);
    \draw[<->, line width=1](2, -1) -- node[above] {$y$} (4, -1);
    
    \draw[<->, line width=1](0, -1.5) -- node[above] {$x$} (4, -1.5);    

\end{tikzpicture}
\caption{BWMCL: diagrammatic representation of the lumped element model (the left part of the diagram depicts the linear viscous energy dissipation element, the right part of the diagram depicts the BWHCL element)}\label{fig:BWMCL_LEM}
\end{figure}

After the introduction of an additional state variable $w \triangleq \dot{y}$, a feedback interconnection of the abstract collision model given by Eq. \eqref{eq:main} and the BWMCL results in the following model:
\begin{equation}\label{eq:BWMCM}
\begin{cases}
\dot{r} = \alpha \frac{k}{c} \abs{y}^{p - 1} y + \alpha_c \frac{k}{c} \abs{z}^{p - 1} z\\
\dot{y} = w \\
\dot{z} = A w - \beta \abs{z}^{n - 1} z \abs{w} - \gamma \abs{z}^n w \\
\dot{w} = - \frac{c}{m} \dot{r} - \alpha  p \frac{k}{c} \abs{y}^{p - 1} \dot{y} - \alpha_c p \frac{k}{c} \abs{z}^{p - 1} \dot{z} \\
\begin{matrix} r(0) = y(0) = z(0) = 0, & w(0) = -v_0 \end{matrix}
\end{cases}
\end{equation}
The relative position $x$ and the relative velocity $v$ can be recovered by augmenting the model with the output function given by
\begin{equation}
\begin{cases}\label{eq:BWMCM_output}
x = r + y\\
v = \dot{r} + \dot{y}
\end{cases}
\end{equation}
The model given by Eq. \eqref{eq:BWMCM} and Eq. \eqref{eq:BWMCM_output} shall be referred to as the Bouc-Wen-Maxwell Collision Model (BWMCM).\footnote{Sometimes, instead of using an output function, it may be more convenient to augment the BWMCM with the additional states $x$ and $v$, the equations $\dot{x} = v$ and $\dot{v} = -(c/m) \dot{r}$, and the initial conditions $x(0) = 0$ and $v(0) = -v_0$.}

The BWMCM will now be nondimensionalized.\footnote{See \cite{logan_applied_2013} for a description of the methodology that was employed for the nondimensionalization of the BWMCM.} The relevant fundamental dimensions are mass $M$, length $L$, and time $T$. The dimensions of the state variables are given by $[r] = L$, $[y] = L$, $[z] = L$, and $[w] = L T^{-1}$. The dimensions of the output variables are $[x] = L$ and $[v] = L T^{-1}$. The parameters $A$, $\alpha$, $n$, and $p$ are dimensionless. The dimensions of the remaining parameters are given by $[m] = M$, $[k] = M L^{1 - p} T^{-2}$, $[c] = M T^{-1}$, $[\beta] = L^{-n}$, $[\gamma] = L^{-n}$, and $[v_0] = L T^{-1}$. The nondimensionalized model will evolve with respect to the nondimensionalized time variable $T \triangleq t/T_c$, with the time scale $T_c \in \mathbb{R}_{> 0}$ given by
\begin{equation}\label{eq:BWMCM_T_c}
T_c  \triangleq \left(\frac{1}{\alpha + \alpha_c A^p}\right)^{\frac{1}{p + 1}} \left(\frac{m}{k}\right)^{\frac{1}{p + 1}} v_0^{-\frac{p - 1}{p + 1}}
\end{equation}
The nondimensionalized state variables are given by $R \triangleq r/X_c$, $Y \triangleq y/X_c$, $Z \triangleq z/Z_c$, $W \triangleq w/(X_c/T_c)$, and the nondimensionalized output variables are given by $X \triangleq x/X_c$ and $V \triangleq v/(X_c/T_c)$, with the spatial scales $X_c, Z_c \in \mathbb{R}_{>0}$ given by
\begin{equation}\label{eq:BWMCM_X_c}
X_c \triangleq \left(\frac{1}{\alpha + \alpha_c A^p}\right)^{\frac{1}{p+1}} \left(\frac{m}{k}\right)^{\frac{1}{p + 1}} v_0^{\frac{2}{p+1}}
\end{equation}
\begin{equation}\label{eq:BWMCM_Z_c}
Z_c \triangleq \left(\frac{1}{\alpha + \alpha_c A^p}\right)^{\frac{1}{p+1}} A \left(\frac{m}{k}\right)^{\frac{1}{p + 1}} v_0^{\frac{2}{p + 1}}
\end{equation}
respectively. Introduction of the dimensionless parameters 
\begin{equation}
B \triangleq \left(\frac{A^{p+1}}{\alpha + \alpha_c A^p}\right)^{\frac{n}{p+1}} \frac{\beta}{A} \left(\frac{m}{k}\right)^{\frac{n}{p + 1}} v_0^{\frac{2 n}{p + 1}}  
\end{equation}
\begin{equation}
\mathit{\Gamma} \triangleq \left(\frac{A^{p + 1}}{\alpha + \alpha_c A^p}\right)^{\frac{n}{p + 1}} \frac{\gamma}{A} \left(\frac{m}{k}\right)^{\frac{n}{p + 1}} v_0^{\frac{2 n}{p + 1}}  
\end{equation} 
\begin{equation}
\kappa \triangleq \frac{\alpha }{\alpha + \alpha_c A^p}
\end{equation}
\begin{equation}
\sigma \triangleq \left(\alpha + \alpha_c A^p\right)^{\frac{1}{p + 1}} \frac{1}{c} \left( m^p k \right)^{\frac{1}{p+1}} v_0^{\frac{p - 1}{p + 1}}
\end{equation}
and the abbreviation $\kappa_c \triangleq 1 - \kappa$, and nondimensionalization of Eq. \eqref{eq:BWMCM} and Eq. \eqref{eq:BWMCM_output} results in the model given by
\begin{equation}\label{eq:BWMCM_nd}
\begin{cases}
\dot{R} = \kappa \sigma \abs{Y}^{p - 1} Y + \kappa_c \sigma \abs{Z}^{p - 1} Z\\
\dot{Y} = W \\
\dot{Z} = W - B \abs{Z}^{n - 1} Z \abs{W} - \mathit{\Gamma} \abs{Z}^n W \\
\dot{W} = - \frac{1}{\sigma} \dot{R} - \kappa p \sigma \abs{Y}^{p - 1} \dot{Y} - \kappa_c p \sigma \abs{Z}^{p - 1} \dot{Z} \\
\begin{matrix} R(0) = Y(0) = Z(0) = 0, & W(0) = -1 \end{matrix}
\end{cases}
\end{equation}
and
\begin{equation}\label{eq:BWMCM_nd_output}
\begin{cases}
X = R + Y\\
V = \dot{R} + \dot{Y}
\end{cases}
\end{equation}
This form of the collision model shall be referred to as the Nondimensionalized Bouc-Wen-Maxwell Collision Model (NDBWMCM).\footnote{Sometimes, instead of using an output function, it may be more convenient to augment the NDBWMCM with the additional states $X$ and $V$, the equations $\dot{X} = V$ and $\dot{V} = -(1/\sigma) \dot{R}$, and the initial conditions $X(0) = 0$ and $V(0) = -1$.}

Under the assumption that the NDBWMCM is parameterized by $B \in \mathbb{R}_{> 0}$, $\mathit{\Gamma} \in (-B, B)$, $\kappa \in (0, 1)$, $\sigma \in \mathbb{R}_{>0}$, $n \in \mathbb{R}_{\geq 1}$, $p \in \mathbb{R}_{\geq 2} \cup \{ 1 \}$, there exists a unique bounded solution of the NDBWMCM on any time interval $[0, T_e)$ with $T_e \in \mathbb{R}_{>0} \cup \{ +\infty \}$. Moreover, the output associated with this solution is bounded. The set of equilibrium points of the NDBWMCM is
\begin{equation}
\mathcal{E} \triangleq \left\{ \left( R, Y, - \left( \frac{\kappa}{\kappa_c}\right)^{\frac{1}{p}} Y, 0 \right) : R, Y  \in \mathbb{R} \right\}
\end{equation}
Each solution of the NDBWMCM converges to a subset of $\mathcal{E}$ at a finite distance from the origin. See Appendix \ref{sec:analysis_BWMCM} for further details.\footnote{As previously (see Section \ref{sec:KVCL}), the conditions that are imposed on the parameters are sufficient, not necessary.} 

The following definitions are similar to the ones provided in Section \ref{sec:KVCL} for the NDBWSHCCM. $\mathcal{P} : \mathbb{P}^{*} \times \mathbb{R}_{>0} \longrightarrow \mathbb{P}$ shall map $(B_b, \mathit{\Gamma}_b, \kappa, \sigma_b, n, p) \in \mathbb{P}^{*}$ and $v_0 \in \mathbb{R}_{> 0}$ to 
\[
\left( B_b v_0^{\frac{2 n}{p + 1}}, \mathit{\Gamma}_b v_0^{\frac{2 n}{p + 1}}, \kappa, \sigma_b v_0^{\frac{p - 1}{p + 1}}, n, p \right) \in \mathbb{P}
\] 
where $\mathbb{P}^{*} = \mathbb{P} \subseteq \mathbb{R}^6$ consist of all $P = (B, \mathit{\Gamma}, \kappa, \sigma, n, p)$ such that $B \in \mathbb{R}_{> 0}$, $\mathit{\Gamma} \in (-B, B)$, $\kappa \in (0, 1)$, $\sigma \in \mathbb{R}_{>0}$, $n \in \mathbb{R}_{\geq 1}$, $p \in \mathbb{R}_{\geq 2} \cup \{ 1 \}$. The elements of $\mathbb{P}^{*}$ shall be referred to as the base parameters of the NDBWMCM. As previously, they are merely convenient abstractions for the study of the behavior of a given physical system represented by the NDBWMCM with respect to the changes in the initial relative velocity (e.g., see Section \ref{sec:ID}).

The function $\Phi : \mathbb{P} \times \mathbb{R}_{\geq 0} \longrightarrow \mathbb{R}^4$ is such that $\Phi_P(T)$ represents the value of the solution of the NDBWMCM parameterized by $P \in \mathbb{P}$ at the time $T \in \mathbb{R}_{\geq 0}$. The function $\Psi : \mathbb{P} \times \mathbb{R}_{\geq 0}\longrightarrow \mathbb{R}^2$ is such that $\Psi_P(T)$ is the value of the output of the NDBWMCM parameterized by $P \in \mathbb{P}$ at the time $T \in \mathbb{R}_{\geq 0}$. The function $F : \mathbb{P} \times \mathbb{R}^4 \longrightarrow \mathbb{R}$ that represents the contact force will be defined as 
\begin{equation}
F_P(R, Y, Z, W) \triangleq - \kappa \abs{Y}^{p - 1} Y - \kappa_c \abs{Z}^{p - 1} Z
\end{equation}
for all $(R, Y, Z, W) \in \mathbb{R}^4$ and $P \in \mathbb{P}$ such that $\kappa = P_3$ and $p = P_6$. With reference to Eq. \eqref{eq:main_t_s}, the time of the separation $T_s : \mathbb{P} \longrightarrow \mathbb{R}_{> 0} \cup \{ +\infty \}$ is given by 
\begin{equation}
T_s (P) \triangleq \inf \{ T \in \mathbb{R}_{\geq 0} : F_P(\Phi_P(T)) \leq 0 \wedge 0 \leq \Psi_{P,2}(T) \}
\end{equation}
for all $P \in \mathbb{P}$. With reference to Eq. \eqref{eq:main_cor}, CoR $e : \mathbb{P} \longrightarrow \mathbb{R}$ is given by  
\begin{equation}
e(P) \triangleq 
\begin{cases}
\Psi_{P,2}(T_s(P)) & T_s(P) \neq +\infty \\
0 & T_s(P) = +\infty
\end{cases}
\end{equation}
for all $P \in \mathbb{P}$.

\section{Parameter Identification}\label{sec:ID}

\begin{figure*}[t]
\begin{subfigure}[t]{0.5\textwidth}
\centering{
    \begin{tikzpicture}
	\definecolor{clr1}{RGB}{150,150,150}
	\definecolor{clr2}{RGB}{0,0,0}
	\begin{axis}[xlabel={$v_0 \: (\text{m} \: \text{s}^{-1})$}, ylabel={$e$}, ymin = 0.5, ymax = 1, xticklabel style={/pgf/number format/precision=2, /pgf/number format/fixed}, legend style={font=\small}, line width=1pt]
	\addplot [clr1, mark=o, mark options={scale=1.5}, only marks] table [x=v_0, y=e, col sep=comma] {data/DSS.csv};
	\addlegendentry{DSS: experiment};
	\addplot [clr1, mark=*, mark options={scale=0.75}, only marks] table [x=v_0, y=e, col sep=comma] {data/DSS_NDBWSHCCM.csv};
	\addlegendentry{DSS: NDBWSHCCM};	
	\addplot [clr2, mark=o, mark options={scale=1.5}, only marks] table [x=v_0, y=e, col sep=comma] {data/DSA.csv};
	\addlegendentry{DSA: experiment};
	\addplot [clr2, mark=*, mark options={scale=0.75}, only marks] table [x=v_0, y=e, col sep=comma] {data/DSA_NDBWSHCCM.csv};
	\addlegendentry{DSA: NDBWSHCCM};
	\end{axis}
	\end{tikzpicture}
}
\subcaption{CoR: NDBWSHCCM vs. experiment}\label{fig:CoR_NDBWSHCCM}
\end{subfigure}
\begin{subfigure}[t]{0.5\textwidth} 
\centering{
    \begin{tikzpicture}
	\definecolor{clr1}{RGB}{150,150,150}
	\definecolor{clr2}{RGB}{0,0,0}    
	\begin{axis}[xlabel={$v_0 \: (\text{m} \: \text{s}^{-1})$}, ylabel={$e$}, ymin = 0.5, ymax = 1, xticklabel style={/pgf/number format/precision=2, /pgf/number format/fixed}, legend style={font=\small}, line width=1pt]
	\addplot [clr1, mark=o, mark options={scale=1.5}, only marks] table [x=v_0, y=e, col sep=comma] {data/DSS.csv};
	\addlegendentry{DSS: experiment};
	\addplot [clr1, mark=*, mark options={scale=0.75}, only marks] table [x=v_0, y=e, col sep=comma] {data/DSS_NDBWMCM.csv};
	\addlegendentry{DSS: NDBWMCM};		
	\addplot [clr2, mark=o, mark options={scale=1.5}, only marks] table [x=v_0, y=e, col sep=comma] {data/DSA.csv};
	\addlegendentry{DSA: experiment};
	\addplot [clr2, mark=*, mark options={scale=0.75}, only marks] table [x=v_0, y=e, col sep=comma] {data/DSA_NDBWMCM.csv};
	\addlegendentry{DSA: NDBWMCM};
	\end{axis}
	\end{tikzpicture}
\subcaption{CoR: NDBWMCM vs. experiment}\label{fig:CoR_NDBWMCM}
}\end{subfigure}
\caption{CoR: models vs. experiment}\label{fig:CoR}
\end{figure*}

A common approach for the parameterization of collision models is identification of model parameters based on experimental CoR data. More specifically, for a given physical system, CoRs are measured for a range of initial relative velocities of the colliding bodies. Then, the model parameters are selected in a manner such that the values of the experimentally obtained CoRs and the values of the CoRs associated with the model are sufficiently close to each other in a certain predefined sense. In what follows, this methodology is applied to the identification of the parameters of the NDBWSHCCM and the NDBWMCM. 

Suppose that the experimental data are provided in the form of a finite sequence of measured relative velocities of the bodies at the time of the collision $\tilde{v}_0 \in \mathbb{R}_{>0}^M$ and a finite sequence of the corresponding measured CoRs $\tilde{e} \in [0, 1]^M$ with $M \in \mathbb{Z}_{\geq 1}$. Then, the quality of the base parameterization $P^{*} \in \mathbb{P}^{*}$ of the NDBWSHCCM or the NDBWMCM can be assessed by the cost function $J : \mathbb{R}_{>0}^M \times [0, 1]^M \times \mathbb{P}^{*} \longrightarrow \mathbb{R}_{\geq 0}$ given by 
\begin{equation}\label{eq:J}
J (\tilde{v}_0, \tilde{e}, P^{*}) \triangleq \sqrt{\sum_{i = 1}^{i = M} (\tilde{e}_i - e (\mathcal{P}(P^{*}, \tilde{v}_{0,i})))^2}
\end{equation}
The cost function can be used in conjunction with global optimization routines to infer the model parameters from experimental data automatically or in conjunction with local optimization routines to refine the model parameters from an initial guess.

The remainder of this section describes an application example based on the experimental data sets provided in Fig. 1 in \cite{kharaz_study_2000}:
\begin{compactitem}
\item ``dataset steel'' (DSS): CoR vs. initial relative velocity for the normal impact of a 5mm diameter aluminum oxide sphere on a thick EN9 steel plate.
\item ``dataset aluminum'' (DSA): CoR vs. initial relative velocity for the normal impact of a 5mm diameter aluminum oxide sphere on a thick aluminum alloy plate.
\end{compactitem}
The data were extracted from \cite{kharaz_study_2000} using the image processing software WebPlotDigitizer \cite{rohatgi_webplotdigitizer_nodate}.

The numerical simulation and the data analysis that are described in this section were performed using Python 3.11.0, NumPy 1.24.2 \cite{harris_array_2020}, and SciPy 1.14.0 \cite{virtanen_scipy_2020}, and relied on the IEEE-754 floating point arithmetic (with the default rounding mode) for the quantization of real numbers \cite{ieee_ieee_2019}. The code is available from the personal repository of the corresponding author.\footnote{\url{https://gitlab.com/user9716869/BWBCL}}

All numerical simulations were performed using the explicit Runge-Kutta method of order 8 \cite{dormand_family_1980, prince_high_1981, hairer_solving_1993} available via the interface of the function \texttt{integrate.solve\_ivp} from the library SciPy 1.14.0 \cite{virtanen_scipy_2020}. All settings of \texttt{integrate.solve\_ivp} were left at their default values, with the exception of the maximum time step (\texttt{max\_step}), the relative tolerance (\texttt{rtol}), and the absolute tolerance (\texttt{atol}). The maximum time step was set to $\approx 10^{-2}$, the relative tolerance was set to $\approx 10^{-10}$, and the absolute tolerance was set to $\approx 10^{-12}$.

The parameter identification was performed using an implementation of the Nelder-Mead algorithm \cite{nelder_simplex_1965} available via the interface of the SciPy function \texttt{optimize.minimize}. The details of the parameter identification process were deemed unimportant and will not be described in the article. The approximations of the values of the identified parameters and the associated values of the cost function are shown in Table \ref{tab:param}. Figure \ref{fig:CoR_NDBWSHCCM} shows the plots of CoR against the initial relative velocity obtained experimentally and from the results of the numerical simulations of the NDBWSHCCM. Figure \ref{fig:CoR_NDBWMCM} shows the plots of CoR against the initial relative velocity obtained experimentally and from the results of the numerical simulations of the NDBWMCM. 

\begin{table}[t]
\caption{DSS and DSA: parameter identification\label{tab:param}}
\centering{
\begin{tabular}{c c c c c}
\toprule
& \multicolumn{2}{c}{NDBWSHCCM} & \multicolumn{2}{c}{NDBWMCM}\\
\midrule
$P^{*}$ and $J$ & DSS & DSA & DSS & DSA \\
\midrule
$B_b$             & $1.43$    & $0.63$    & $0.655$  & $0.44$ \\
$\mathit{\Gamma}_b$ & $-1.42$   & $-0.611$  & $-0.64$  & $-0.418$ \\
$\kappa$          & $0.632$   & $0.188$   & $0.519$  & $0.113$ \\
$\sigma_b$          & $0.00715$ & $0.00594$ & $0.0118$ & $0.00785$ \\
$n$               & $1.31$    & $1$       & $1.94$   & $1.27$ \\
$p$               & $1.27$    & $2.02$    & $2.28$   & $3.14$ \\
\midrule
$J (\tilde{v}_0, \tilde{e}, P^{*})$ & $0.0357$ & $0.011$ & $0.0406$ & $0.0114$ \\
\bottomrule
\end{tabular}
}
\end{table}

\begin{figure}[t]
\centering
\begin{tikzpicture}
\definecolor{clr1}{RGB}{200,200,200}
\definecolor{clr2}{RGB}{150,150,150}
\definecolor{clr3}{RGB}{100,100,100}
\definecolor{clr4}{RGB}{0,0,0}
\begin{axis}[xlabel={$\abs{x} \: (\text{mm})$}, ylabel={$F \: (\text{N})$}, ymin = 0, xticklabel style={/pgf/number format/precision=2, /pgf/number format/fixed}, legend pos=north west, legend style={font=\small}, line width=1pt]
\addplot [clr1, mark=none, dashed] table [x=x, y=F, col sep=comma] {data/RCross_2_15.csv};
\addlegendentry{$v_0 = 2.15 \: (\text{m} \: \text{s}^{-1})$};
\addplot [clr2, mark=none, dashed] table [x=x, y=F, col sep=comma] {data/RCross_3_03.csv};
\addlegendentry{$v_0 = 3.03 \: (\text{m} \: \text{s}^{-1})$};
\addplot [clr3, mark=none, dashed] table [x=x, y=F, col sep=comma] {data/RCross_4_18.csv};
\addlegendentry{$v_0 = 4.18 \: (\text{m} \: \text{s}^{-1})$};
\addplot [clr4, mark=none, dashed] table [x=x, y=F, col sep=comma] {data/RCross_5_02.csv};
\addlegendentry{$v_0 = 5.02 \: (\text{m} \: \text{s}^{-1})$};
\addplot [clr1, mark=none, mark options={scale=0.75}] table [x expr=1000*\thisrow{x}, y=F, col sep=comma] {data/RCross_sim_2_15.csv};
\addlegendentry{$v_0 = 2.15 \: (\text{m} \: \text{s}^{-1})$};
\addplot [clr2, mark=none, mark options={scale=0.75}] table [x expr=1000*\thisrow{x}, y=F, col sep=comma] {data/RCross_sim_3_03.csv};
\addlegendentry{$v_0 = 3.03 \: (\text{m} \: \text{s}^{-1})$};
\addplot [clr3, mark=none, mark options={scale=0.75}] table [x expr=1000*\thisrow{x}, y=F, col sep=comma] {data/RCross_sim_4_18.csv};
\addlegendentry{$v_0 = 4.18 \: (\text{m} \: \text{s}^{-1})$};
\addplot [clr4, mark=none, mark options={scale=0.75}] table [x expr=1000*\thisrow{x}, y=F, col sep=comma] {data/RCross_sim_5_02.csv};
\addlegendentry{$v_0 = 5.02 \: (\text{m} \: \text{s}^{-1})$};
\end{axis}
\end{tikzpicture}
\caption{Normal impact of a baseball on a flat surface: experimentally obtained hysteresis loops (dashed lines) vs. hysteresis loops obtained from the numerical simulations of the BWSHCCM (solid lines) \label{fig:RCross}}
\end{figure}

\begin{table}[t]
\caption{Normal impact of a baseball on a flat surface: parameterization of the BWSHCCM\label{tab:RCross}}
\centering{
\begin{tabular}{c c c c c}
\toprule
Parameter & Value & Unit \\
\midrule
$m$ & $0.146$ & $\text{kg}$ \\
$k$ & $117080063$ & $\text{kg} \: \text{m}^{1 - p} \: \text{s}^{-2}$ \\
$c$ & $5854003$ & $\text{kg} \: \text{m}^{-p} \: \text{s}^{-1}$ \\
$n$ & $1.1$ & $\text{-}$ \\
$p$ & $1.7$ & $\text{-}$ \\
$\alpha$ & $0.1$ & $\text{-}$ \\
$\beta$ & $981.05$ & $\text{m}^{-n}$ \\
$\gamma$ & $-961.4$ & $\text{m}^{-n}$ \\
$A$ & $0.925$ & $\text{-}$ \\
\bottomrule
\end{tabular}
}
\end{table}

For each data set, the results indicate good agreement between the CoR data obtained experimentally and the CoR data obtained from the results of the numerical simulations of the models across a wide band (low to moderate) of the initial relative velocities. It is important to note that only a single base parameter vector $P^{*} \in \mathbb{P}^{*}$ was employed for each data set. Therefore, only a single vector of physical parameters $(m, k, c, n, p, \alpha, \beta, \gamma, A)$ is needed to achieve a good agreement between the CoR data obtained from the experiments and the CoR data stemming from the numerical simulations of the models. 

It should be noted that a preliminary informal parameter sensitivity analysis that was performed by the authors suggests that the parameter identification based on the CoR data alone may not provide a unique vector of parameters (at least in a statistical sense) for either of the nondimensionalized models: multiple statistically indistinguishable solutions may be possible. Therefore, the cost function given by Eq. \eqref{eq:J} may be augmented to penalize further optimization criteria (e.g., duration of the contact), if such data are available. Alternatively, the problem can be reformulated as a multi-objective optimization problem. 

It is also possible to identify the parameters of the models based on the experimentally obtained time domain data or the associated hysteresis loops. Figure \ref{fig:RCross} shows the plots of the experimentally obtained hysteresis loops observed during the normal impact of a baseball on a flat surface across a range of initial relative velocities. The experimental data were provided by Professor Rodney Cross: the data originally appeared in Fig. 9.5 in \cite{cross_physics_2011} (see also \cite{cross_impact_2014}). The same figure shows the hysteresis loops obtained based on the results of the numerical simulations of the BWSHCCM with the parameters shown in Table \ref{tab:RCross}. The settings for the numerical simulations were identical to the settings used in the identification study based on the CoR data, with the exception of the maximum time step, which was set to $\approx T_c/100 \: \text{s}$. The plots demonstrate a good agreement between the experimentally obtained hysteresis loops and the hysteresis loops obtained from the simulation of the BWSHCCM. The parameter identification study was based on an informal procedure. As previously, its details were deemed unimportant and will not be described in the article.

\section{Conclusions and Future Work}\label{sec:conclusions}

The article showcased an analytical and numerical study of two mathematical models of binary direct collinear collisions of convex viscoplastic bodies. The mathematical models of the collision process employed two distinct incremental collision laws based on the Bouc-Wen differential model of hysteresis. It was demonstrated that the models possess favorable analytical properties (e.g., global existence, uniqueness, and boundedness of the solutions) under mild restrictions on the values of the model parameters. Two model parameter identification strategies were proposed and tested using experimental collision data available in the research literature. Based on the results of the identification studies it was concluded that one set of model parameters independent of the initial relative velocity is sufficient for attainment of a good correlation between the results of the numerical simulations of the models under consideration in this study and the experimental data. Therefore, it can be concluded that the models accurately describe the physics of a wide variety of contact and collision phenomena.

Possible future directions may include:
\begin{compactitem}
\item Extensions of the collision laws to account for the details of the underlying physical phenomenon (e.g., introduction of the distinct elastic loading and elastic-plastic loading stages).
\item Investigation of the BWSHCCL and BWMCL in the context of multiple simultaneous collisions (e.g., see \cite{brogliato_nonsmooth_2016, stronge_impact_2018}).
\item Comparative analysis of alternative differential models of hysteresis in the context of impact dynamics.
\item Investigation of planar and three-dimensional collisions of viscoplastic bodies with rough surfaces based on the Bouc-Wen model of hysteresis.
\item Experimental studies that could help to understand the limitations of the BWSHCCL and the BWMCL.
\item Construction of analytical approximations of the solutions of the IVPs associated with the BWSHCCM and the BWMCM.
\item Applications of the BWSHCCL and the BWMCL to problems of practical significance and comparative analysis of the BWSHCCL and the
BWMCL with other collision laws in the context of applications.
\item A study of the applicability of the BWSHCCL and the BWMCL to impacts with elastic aftereffect (e.g., see \cite{wang_advanced_2017}).
\item Development of a methodology for parameterization of the BWSHCCL and the BWMCL from first principles, based on the properties of the materials and the geometry of the colliding bodies (without relying on the model parameter identification).
\end{compactitem}

\section*{Acknowledgment} 

The authors would like to acknowledge their families, colleagues, and friends. Special thanks go to two anonymous reviewers. Multiple significant improvements were introduced to the original draft of the article based on their feedback. Special thanks also go to Professor Rodney Cross for providing experimental data from \cite{cross_physics_2011}. Special thanks also go to the members of staff of Auburn University Libraries for their assistance in finding rare and out-of-print research articles and research monographs. The authors would also like to acknowledge the professional online communities, instructional websites, and various online service providers, especially \url{https://www.adobe.com/acrobat/online/pdf-to-word.html}, \url{https://automeris.io}, \url{https://capitalizemytitle.com}, \url{https://www.matweb.com}, \url{https://www.overleaf.com}, \url{https://pgfplots.net}, \url{https://scholar.google.com}, \url{https://stackexchange.com}, \url{https://stringtranslate.com}, \url{https://www.wikipedia.org}. We also note that the results of some of the calculations that are presented in this article were performed with the assistance of the software Wolfram Mathematica \cite{wolfram_research_inc_mathematica_2023}. Furthermore, the software MATLAB R2023a \cite{mathworks_matlab_2023} was used extensively for numerical experiments. Other software that was used to produce this article included Adobe Acrobat Reader, Adobe Digital Editions, ChatGPT (ChatGPT was used only to review the article and the associated code during the final stages of the preparation of the manuscript in February 2025), DiffMerge, Git, GitLab, Google Chrome, Grammarly (the use of Grammarly was restricted to the identification and correction of spelling, grammar, and punctuation errors), Jupyter Notebook, LibreOffice, macOS Monterey, Mamba, Microsoft Outlook, Preview, Safari, TeX Live/MacTeX, Texmaker, and Zotero.

\section*{Funding Data}

The present work did not receive any specific funding. However, the researchers receive financial support from Auburn University for their overall research activity.

\appendix

\section{Notation and Conventions}\label{sec:notation}

Essentially all of the definitions and results that are employed in this article are standard in the fields of set theory, general topology, analysis, ordinary differential equations, and nonlinear systems/control. They can be found in a number of textbooks and monographs on these subjects (e.g., see \cite{takeuti_introduction_1982}, \cite{kelley_general_1955, morris_topology_2020, baldwin_math_2024}, \cite{bloch_real_2010, shurman_calculus_2016, ziemer_modern_2017}, \cite{chicone_ordinary_1999, schaeffer_ordinary_2016}, \cite{lasalle_extensions_1960, yoshizawa_stability_1966, yoshizawa_stability_1975, lakshmikantham_practical_1990, isidori_nonlinear_1995, sontag_mathematical_1998, isidori_nonlinear_1999, sastry_nonlinear_1999, marquez_nonlinear_2003, haddad_nonlinear_2011, khalil_nonlinear_2015, goebel_set-valued_2024}, respectively). 

\begin{definition}
$\in$ denotes the set membership relation, $\subseteq$ denotes the subset relation, $\subset$ denotes the proper subset relation, $\cup$ denotes the binary set union operation, $\cap$ denotes the binary set intersection operation, $\setminus$ denotes the binary set difference operation, $\mathcal{P}$ denotes the power set operation, $\emptyset$ denotes the empty set, $(a_1, \ldots, a_n)$ denotes an $n$-tuple, $\{ a_1, \ldots, a_n \}$ denotes an unordered collection of elements.\footnote{It should be noted that some of the syntactic constructions may carry different semantics depending on the context. For example, $(a, b)$ may be used as a pair or as an interval. It is hoped that the context of the discussion will always make the meaning of a given syntactic construction apparent.}
\end{definition}

\begin{definition}
By convention, a topological space cannot be empty. Suppose $X \neq \emptyset$ and $\tau \subseteq \mathcal{P} X$ is a topology on $X$. $\mathsf{cl} A$ denotes the closure of $A \subseteq X$; if $Y \subseteq X$ and $Y \neq \emptyset$, then $\tau | Y$ will denote the subspace topology of $\tau$ on $Y$; the sets $A \subseteq X$ and $B \subseteq X$ are separated if and only if $\mathsf{cl} A \cap B = A \cap \mathsf{cl} B = \emptyset$; a set $C \subseteq X$ is clopen if and only if it is open and closed; $A \subseteq X$ is connected if and only if it is not a union of two nonempty separated sets; $(X, \tau)$ is a connected topological space if and only if $X$ is a connected set.  
\end{definition}

It should be noted that different definitions of a connected set and a connected topological space are employed in some of the cited literature. The following technical lemmas establish a connection between the two commonly used definitions (these results are not used directly, and the proofs were deemed to be sufficiently simple to be omitted):
\begin{lem}
Suppose $(X, \tau)$ is a topological space. Then, $(X, \tau)$ is connected if and only if the only clopen sets in $(X, \tau)$ are $\emptyset$ and $X$.
\end{lem}
\begin{lem}\label{thm:connected_set_alt_def}
Suppose $(X, \tau)$ is a topological space and $Y \subseteq X$. Then, $Y$ is a connected set in $(X, \tau)$ if and only if either $Y = \emptyset$ or $(Y, \tau | Y)$ is a connected topological space.
\end{lem}

The following technical lemma will be employed in the proofs of several results that follow (the proof was deemed to be sufficiently simple to be omitted):
\begin{lem}\label{thm:con_sep}
Suppose $(X, \tau)$ is a topological space. Suppose that $A \subseteq X$ and $B \subseteq X$ are separated, $C \subseteq A \cup B$ is connected. Then, $C \subseteq A$ or $C \subseteq B$. 
\end{lem}

\begin{definition}
$\mathbb{Z}$ is the set of all integers; $\mathbb{R}$ is the set of all real numbers; an interval of real numbers $I \subseteq \mathbb{R}$ is non-degenerate if it has a non-empty interior; $\mathbb{K}_{>a} \triangleq (a, +\infty) \cap \mathbb{K}$, $\mathbb{K}_{<a} \triangleq (-\infty, a) \cap \mathbb{K}$, $\mathbb{K}_{\geq a} \triangleq [a, +\infty) \cap \mathbb{K}$, and $\mathbb{K}_{\leq a} \triangleq (-\infty, a] \cap \mathbb{K}$ for any $a \in \mathbb{R}$ with $\mathbb{K} \subseteq \mathbb{R}$; $\mathbb{R}^n$ with $n \in \mathbb{Z}_{\geq 1}$ is the set of $n$-tuples of real numbers (augmented with the structure of the Euclidean space); if $X = (x_1, \ldots, x_n) \in \mathbb{R}^n$ with $n \in \mathbb{Z}_{\geq 1}$, then $X_i \triangleq x_i$ for all $i \in \{ 1, \ldots, n \}$; $f : X \longrightarrow Y$ denotes a function with the domain $X$ and the codomain $Y$; given $f : X \longrightarrow Y$, $f(A)$ denotes the image of $f$ under the set $A$; if $f : X \longrightarrow \mathbb{R}^n$ with $n \in \mathbb{Z}_{\geq 1}$, then $f_i : X \longrightarrow \mathbb{R}$ is given by $f_i(x) \triangleq (f(x))_i$ for all $x \in X$ and $i \in \{ 1, \ldots, n \}$; unless stated otherwise, the topology of a subset of $\mathbb{R}^n$ with $n \in \mathbb{Z}_{\geq 1}$ is always the subspace topology of the standard topology on $\mathbb{R}^n$; given $A \subseteq \mathbb{R}$, $\inf A \in \mathbb{R} \cup \{ -\infty, +\infty \}$ denotes the infimum of $A$ and $\sup A \in \mathbb{R} \cup \{ -\infty, +\infty \}$ denotes the supremum of $A$; given a sequence $\{ x_i \in \mathbb{R}^n \}_{i \in \mathbb{Z}_{\geq 1}}$ with $n \in \mathbb{Z}_{\geq 1}$, $\lim_{i \rightarrow +\infty} x_i$ denotes the limit of $x$, provided that it exists; $\langle \cdot, \cdot \rangle : \mathbb{R}^n \times \mathbb{R}^n \longrightarrow \mathbb{R}$ with $n \in \mathbb{Z}_{\geq 1}$ is the canonical inner product on $\mathbb{R}^n$; $\lVert \cdot \rVert_2 : \mathbb{R}^n \longrightarrow \mathbb{R}_{\geq 0}$ with $n \in \mathbb{Z}_{\geq 1}$ is the Euclidean norm on $\mathbb{R}^n$; assuming that $n \in \mathbb{Z}_{\geq 1}$, $a \in \mathbb{R}^n$, and $r \in \mathbb{R}_{>0}$, $\mathbb{B}(a, r) \triangleq \{\ x \in \mathbb{R}^n : \lVert x - a \rVert_2 < r \}$ is an open ball in $\mathbb{R}^n$ centered at $a$ with the radius $r$; assuming that $n \in \mathbb{Z}_{\geq 1}$, $a \in \mathbb{R}^n$, and $r \in \mathbb{R}_{>0}$, $\bar{\mathbb{B}}(a, r) \triangleq \{\ x \in \mathbb{R}^n : \lVert x - a \rVert_2 \leq r \}$ is a closed ball in $\mathbb{R}^n$ centered at $a$ with the radius $r$; $C \subseteq \mathbb{R}^n$ is convex if and only if $(1 - \lambda) x + \lambda y \in C$ for all $x, y \in C$ and for all $\lambda \in [0, 1]$; $C \subseteq \mathbb{R}^n$ is strictly convex if and only if $(1 - \lambda) x + \lambda y$ belongs to the interior of $C$ for all $x, y \in C$ and for all $\lambda \in (0, 1)$; $f : \mathbb{R}^n \longrightarrow \mathbb{R}^n$ with $n \in \mathbb{Z}_{\geq 1}$ is locally Lipschitz if and only if for every $x \in \mathbb{R}^n$ there exists an open set $U \subseteq \mathbb{R}^n$ such that $x \in U$ and there exists $L \in \mathbb{R}_{>0}$ such that $\lVert f(y) - f(z) \rVert_2 \leq L \lVert y - z \rVert_2$ for all $y, z \in U$; given a differentiable function $f : X \longrightarrow Y$ such that $X \subseteq \mathbb{R}$ and $Y \subseteq \mathbb{R}^n$ with $n \in \mathbb{Z}_{\geq 1}$, $df/dx$ denotes the derivative of the function; the overdot notation $\dot{x} \triangleq (dx/dt)$ represents the derivative of the differentiable function $x$ with respect to the time variable (in the context of mechanics); given a continuously differentiable function $f : \mathbb{R}^n \longrightarrow \mathbb{R}$, $\nabla f : \mathbb{R}^n \longrightarrow \mathbb{R}^n$ represents the gradient of $f$.  
\end{definition}

\begin{definition}
Consider the following system of ordinary differential equations with an output
\begin{equation}\label{eq:sys}
\begin{cases}
\dot{x} = f(x)\\
y = g(x)
\end{cases}
\end{equation}
where $f : \mathbb{R}^n \longrightarrow \mathbb{R}^n$ with $n \in \mathbb{Z}_{\geq 1}$ is a locally Lipschitz continuous state function, and $g : \mathbb{R}^n \longrightarrow \mathbb{R}^k$ with $k \in \mathbb{Z}_{\geq 1}$ is a continuous output function. Equation \eqref{eq:sys} augmented with an initial condition $x(0) = x_0 \in \mathbb{R}^n$ shall be referred to as an initial value problem (IVP) associated with the system given by Eq. \eqref{eq:sys}. A differentiable function $x : I \longrightarrow \mathbb{R}^n$ with $I \subseteq \mathbb{R}$ being a non-degenerate interval such that $0 \in I$ is a solution of the IVP associated with the system given by Eq. \eqref{eq:sys} with the initial condition $x_0 \in \mathbb{R}^n$ if $x(0) = x_0$ and $\dot{x}(t) = f(x(t))$ for all $t \in I$. 
\end{definition}

The following definitions were adopted from \cite{bhat_nontangency-based_2003} and \cite{haddad_nonlinear_2011}:
\begin{definition}
For the remainder of this definition, suppose that the system given by Eq. \eqref{eq:sys} has a unique solution defined on $\mathbb{R}_{\geq 0}$ for every initial condition. Suppose that $x : \mathbb{R}_{\geq 0} \longrightarrow \mathbb{R}^n$ is a solution of an IVP associated with the system given by Eq. \eqref{eq:sys} with the initial condition $x(0) = z \in \mathbb{R}^n$. Then, $\mathcal{O}_z^{+} \triangleq \{ x(t) : t \in \mathbb{R}_{\geq 0} \}$ is the positive orbit of $z$. A set $U \subseteq \mathbb{R}^n$ is positively invariant with respect to the system given by Eq. \eqref{eq:sys} if and only if for every solution $x : \mathbb{R}_{\geq 0} \longrightarrow \mathbb{R}^n$ of the IVP with $x(0) = z \in U$, $x(t) \in U$ for all $t \in \mathbb{R}_{\geq 0}$. A set $U \subseteq \mathbb{R}^n$ is negatively invariant with respect to the system given by Eq. \eqref{eq:sys} if and only if for every $z \in U$ and $T \in \mathbb{R}_{\geq 0}$ there exists a solution $x : [0, T] \longrightarrow U$ of the IVP with $x(T) = z$. A set $U \subseteq \mathbb{R}^n$ is invariant with respect to the system given by Eq. \eqref{eq:sys} if and only if it is positively invariant and negatively invariant with respect to the system given by Eq. \eqref{eq:sys}. Suppose again that $x : \mathbb{R}_{\geq 0} \longrightarrow \mathbb{R}^n$ is a solution of an IVP associated with the system given by Eq. \eqref{eq:sys} with the initial condition $x(0) = z \in \mathbb{R}^n$. Then, $p \in \mathbb{R}^n$ is a positive limit point of $z$ if and only if there exists a nondecreasing sequence $\{ t_n \}_{n \in \mathbb{Z}_{\geq 1}}$ of positive real numbers such that $\lim_{n \rightarrow +\infty} t_n = +\infty$ and $\lim_{n \rightarrow +\infty} x(t_n) = p$. Furthermore, $\mathcal{O}_z^{+\infty} \subseteq \mathbb{R}^n$ shall be used to denote the positive limit set of $z$, that is, the set of all positive limit points of $z$. $\lim_{t \rightarrow +\infty} x (t) = A \subseteq \mathbb{R}^n$ if and only if for every $\varepsilon \in \mathbb{R}_{>0}$ there exists $T \in \mathbb{R}_{>0}$ such that $\inf_{p \in A} \lVert x(t) - p \rVert_2 < \varepsilon$ for all $t > T$. A continuous and strictly increasing function $\alpha : \mathbb{R}_{\geq 0} \longrightarrow \mathbb{R}_{\geq 0}$ is of class $\mathcal{K}_{\infty}$ if and only if $\alpha(0) = 0$ and $\lim_{x \rightarrow +\infty} \alpha (x) = +\infty$.
\end{definition}
The following definition was adopted from \cite{yoshizawa_stability_1975}:
\begin{definition}
The solutions of the system given by Eq. \eqref{eq:sys} are said to be equi-bounded if and only if for all $\alpha \in \mathbb{R}_{>0}$ there exists $\beta \in \mathbb{R}_{>0}$ such that $\lVert x(t) \rVert_2 < \beta$ for all $t \in [0, T)$ for every solution $x : [0, T) \longrightarrow \mathbb{R}^n$ with $T \in \mathbb{R}_{>0} \cup \{ +\infty \}$ starting from the initial condition $x(0) = x_0 \in \mathbb{R}^n$ such that $\lVert x_0 \rVert_2 \leq \alpha$.
\end{definition}
The following definition extends the concept of equi-boundedness to a system with an output:
\begin{definition}
The outputs of the system given by Eq. \eqref{eq:sys} are said to be equi-bounded if and only if for all $\alpha \in \mathbb{R}_{>0}$ there exists $\gamma \in \mathbb{R}_{>0}$ such that $\lVert y(t) \rVert_2 < \gamma$ for all $t \in [0, T)$ for every output $y : [0, T) \longrightarrow \mathbb{R}^n$ with $T \in \mathbb{R}_{>0} \cup \{ +\infty \}$ that corresponds to a solution $x : [0, T) \longrightarrow \mathbb{R}^n$ that starts from the initial condition $x(0) = x_0 \in \mathbb{R}^n$ such that $\lVert x_0 \rVert_2 \leq \alpha$.
\end{definition}
The following technical lemma showcases that the equi-boundedness of the solutions of the system associated with Eq. \eqref{eq:sys} implies the equi-boundedness of its outputs:
\begin{lem}\label{thm:output_bound}
Suppose that the solutions of the system given by Eq. \eqref{eq:sys} are equi-bounded. Then, the outputs of the system given by Eq. \eqref{eq:sys} are equi-bounded. 
\end{lem}
\begin{proof}
Fix $\alpha \in \mathbb{R}_{>0}$. Since the solutions are equi-bounded, obtain $\beta \in \mathbb{R}_{>0}$ such that $x([0, T)) \subseteq B \triangleq \bar{\mathbb{B}}(0, \beta)$ for every solution $x : [0, T) \longrightarrow \mathbb{R}^n$ with $T \in \mathbb{R}_{>0} \cup \{ +\infty \}$ starting from the initial condition $x(0) = x_0 \in \mathbb{R}^n$ such that $\lVert x_0 \rVert_2 \leq \alpha$. Since $B$ is a compact set and $g$ is a continuous function defined on a superset of $B$, $\lVert g(x) \rVert_2 < \gamma$ for all $x \in B$ for some $\gamma \in \mathbb{R}_{>0}$ as a consequence of the Extreme Value Theorem (e.g., see Theorem 2.4.15 in \cite{shurman_calculus_2016}). Suppose that $y : [0, T) \longrightarrow \mathbb{R}^k$ with $T \in \mathbb{R}_{>0}$ is the output associated with a solution $x : [0, T) \longrightarrow \mathbb{R}^n$ that starts from the initial condition $x(0) = x_0 \in \mathbb{R}^n$ with $\lVert x_0 \rVert_2 \leq \alpha$. Then, $\lVert y(t) \rVert_2 = \lVert g(x(t)) \rVert_2 < \gamma$ for all $t \in [0, T)$. 
\end{proof}

\section{Analysis of the NDBWSHCCM}\label{sec:analysis_BWSHCCM}

Suppose that the NDBWSHCCM is parameterized by $B \in \mathbb{R}_{\geq 0}$, $\mathit{\Gamma} \in [-B, B]$, $\kappa \in (0, 1)$, $\sigma \in \mathbb{R}_{\geq 0}$, $n, p \in \mathbb{R}_{\geq 1}$. Suppose also that $\kappa_c = 1 - \kappa$. Let $f : \mathbb{R}^3 \longrightarrow \mathbb{R}^3$ denote the state function associated with the NDBWSHCCM. It is given by 
\[
\begin{cases}
f_1(\mathbf{X}) \triangleq V \\
f_2(\mathbf{X}) \triangleq V - B \abs{Z}^{n - 1} Z \abs{V} - \mathit{\Gamma} \abs{Z}^n V \\
f_3(\mathbf{X}) \triangleq - \kappa \abs{X}^{p - 1} X - \kappa_c \abs{Z}^{p - 1} Z - \sigma \abs{X}^p V
\end{cases}
\]
for all $\mathbf{X} \triangleq (X, Z, V) \in \mathbb{R}^3$.\footnote{The informal notation $\mathbf{X} \triangleq (A_1, \ldots, A_k)$ will be used to introduce a symbol $\mathbf{X}$ for a vector in $\mathbb{R}^k$ with $k \in \mathbb{Z}_{\geq 1}$ and an additional symbol for each of its components.} The restrictions on the initial conditions that are stated in the main body of the article will be relaxed to $X(0) = X_0$, $Z(0) = Z_0$, $V(0) = V_0$ with $X_0, Z_0, V_0 \in \mathbb{R}$.

Define the Lyapunov candidate $\mathcal{V} : \mathbb{R}^3 \longrightarrow \mathbb{R}$ as
\[
\mathcal{V}(\mathbf{X}) \triangleq \frac{\kappa}{p + 1} \abs{X}^{p + 1} + \frac{\kappa_c}{p + 1} \abs{Z}^{p + 1} + \frac{1}{2} V^2
\]
for all $\mathbf{X} \triangleq (X, Z, V) \in \mathbb{R}^3$. Note that $\mathcal{V}$ is continuously differentiable and radially unbounded (i.e., $\lim_{\lVert \mathbf{X} \rVert_2 \rightarrow +\infty} \mathcal{V} (\mathbf{X}) = +\infty$), $\mathcal{V}(0) = 0$, and $\mathcal{V}(\mathbf{X}) > 0$ for all $\mathbf{X} \in \mathbb{R}^3 \setminus \{ 0 \}$. Introduce the notation $\dot{\mathcal{V}}(\mathbf{X}) \triangleq \langle \nabla \mathcal{V} (\mathbf{X}), f(\mathbf{X}) \rangle$ and note that
\[
\dot{\mathcal{V}}(\mathbf{X}) = -\sigma \abs{X}^p V^2 - \kappa_c \abs{Z}^{p + n - 1} \left( B \abs{Z} \abs{V} + \mathit{\Gamma} Z V \right)
\]
for all $\mathbf{X} \in \mathbb{R}^3$. Then,
\begin{lem}\label{thm:BWSHCCM_dV_leq_0}
Under the restrictions on the values of the parameters stated above, $\dot{\mathcal{V}}(\mathbf{X}) \leq 0$ for all $\mathbf{X} \in \mathbb{R}^3$. 
\end{lem}
\begin{proof}
Note that $-\sigma \abs{X}^p V^2 \leq 0$ for all $X, V \in \mathbb{R}$. Then, it suffices to show that $- \kappa_c \abs{Z}^{p + n - 1} \left( B \abs{Z} \abs{V} + \mathit{\Gamma} Z V \right) \leq 0$. If $Z = 0$, then the inequality above holds. Otherwise, it suffices to show that $0 \leq B \abs{Z} \abs{V} + \mathit{\Gamma} Z V$. Consider the following cases:
\begin{compactitem}
\item Case I: $Z V \leq 0$. Note that $\mathit{\Gamma} - B \leq 0$. Multiplying both sides by $Z V$ and taking into account that $Z V = - \abs{Z V} = -\abs{Z} \abs{V}$ results in the desired inequality.
\item Case II: $Z V \geq 0$. Note that $0 \leq B + \mathit{\Gamma}$. Multiplying both sides by $Z V$ and taking into account that $Z V = \abs{Z V} = \abs{Z} \abs{V}$ results in the desired inequality.
\end{compactitem}
Therefore, $\dot{\mathcal{V}}(\mathbf{X}) \leq 0$ for all $\mathbf{X} \in \mathbb{R}^3$.
\end{proof}

\begin{prop}\label{thm:BWSHCCM_EUB}
Under the restrictions on the values of the parameters stated above, there exists a unique solution of the NDBWSHCCM on any time interval $[0, T)$ with $T \in \mathbb{R}_{>0} \cup \{ +\infty \}$ for every initial condition $(X_0, Z_0, V_0) \in \mathbb{R}^3$. Furthermore, the solutions of the NDBWSHCCM are equi-bounded.
\end{prop}
\begin{proof}
Note that the state function associated with the NDBWSHCCM is locally Lipschitz and $f(0) = 0$. Note also that, taking into account the properties of $\mathcal{V}$ that were exposed above, there exists $\phi \in \mathcal{K}_{\infty}$ such that $\phi (\lVert \mathbf{X} \rVert_2) \leq \mathcal{V} (\mathbf{X})$ for all $\mathbf{X} \in \mathbb{R}^3$ (e.g., see \cite{hahn_stability_1967, haddad_nonlinear_2011, kellett_compendium_2014}). Lastly, note that $\dot{\mathcal{V}}(\mathbf{X}) \leq 0$ for all $\mathbf{X} \in \mathbb{R}^3$ by Lemma \ref{thm:BWSHCCM_dV_leq_0}. Then, by Theorem 8.7 in \cite{yoshizawa_stability_1975}, the solutions of the IVPs associated with the NDBWSHCCM are equi-bounded. Since the state function associated with the NDBWSHCCM is locally Lipschitz, a unique solution forward in time of the NDBWSHCCM exists on a maximal interval of existence of the form $[0, T)$ with $T \in \mathbb{R}_{>0} \cup \{ +\infty \}$ for every initial condition (e.g., see Theorem 4.1.1 in \cite{schaeffer_ordinary_2016} or Theorem 2.25 in \cite{haddad_nonlinear_2011}). Since each solution defined on a maximal interval of existence is bounded, each solution lies in a compact set. Therefore, by the theorem on the extendability of the solutions (e.g., see Theorem 4.1.2 in \cite{schaeffer_ordinary_2016} or Corollary 2.5 in \cite{haddad_nonlinear_2011}), each solution can be extended to a unique solution on $[0, +\infty)$.
\end{proof}

Define the set
\[
\mathcal{E} \triangleq \left\{ \left(X, - \left(\frac{\kappa}{\kappa_c} \right)^{\frac{1}{p}} X, 0 \right) : X \in \mathbb{R} \right\}
\]
\begin{prop}\label{thm:BWSHCCM_eq}
Under the restrictions on the values of the parameters stated above, $\mathcal{E}$ is the set of all equilibrium points of the NDBWSHCCM.
\end{prop}
\begin{proof}
Suppose that $(X, Z, V) \in \mathbb{R}^3$ is an equilibrium point of the NDBWSHCCM. Then, $V = 0$ and $\abs{Z}^{p - 1} Z = - (\kappa/\kappa_c) \abs{X}^{p - 1} X$. Therefore, taking into account the restrictions on the values of the parameters, it can be shown that $Z = -(\kappa/\kappa_c)^{1/p} X$. Thus, $(X, Z, V) \in \mathcal{E}$. Suppose $(X, Z, V) \in \mathcal{E}$. In this case, $V = 0$ and $Z = - \left(\kappa/\kappa_c \right)^{1/p} X$. Then, it can be verified by substitution that $f \left(X, Z, V \right) = 0$. Thus, $(X, Z, V)$ is, indeed, an equilibrium point of the NDBWSHCCM.
\end{proof}

Define the following sets:
\[
\mathcal{U}_1 \triangleq \{ (0, 0, V) : V \in \mathbb{R} \}
\] 
\[
\mathcal{U}_2 \triangleq \{ (X, Z, 0) : X, Z \in \mathbb{R} \}
\] 
\begin{lem}\label{thm:BWSHCCM_inv_dV_1}
Suppose that the restrictions on the values of the parameters stated above are amended as follows: $B \in \mathbb{R}_{> 0}$, $\mathit{\Gamma} \in (-B, B)$, $\sigma \in \mathbb{R}_{> 0}$. Then, $\dot{\mathcal{V}}^{-1}(0) = \mathcal{U}_1 \cup \mathcal{U}_2$ and the largest invariant set that is contained in $\dot{\mathcal{V}}^{-1}(0)$ is $\mathcal{E}$.
\end{lem}
\begin{proof}
Suppose that $(X, Z, V) \in \dot{\mathcal{V}}^{-1}(0)$. Then, $\dot{\mathcal{V}}(X, Z, V) = 0$. From the proof of Lemma \ref{thm:BWSHCCM_dV_leq_0},
\[
\begin{cases}
\sigma \abs{X}^p V^2 = 0 \\
\kappa_c \abs{Z}^{p + n - 1} \left( B \abs{Z} \abs{V} + \mathit{\Gamma} Z V \right) = 0
\end{cases}
\]
Since $B \neq \mathit{\Gamma}$ and $B \neq -\mathit{\Gamma}$, $B \abs{Z} \abs{V} + \mathit{\Gamma} Z V = 0$ if and only if either $Z = 0$ or $V = 0$. Thus, either $X = Z = 0$ or $V = 0$. That $\dot{\mathcal{V}}(X, Z, V) = 0$ if either $X = Z = 0$ or $V = 0$ can be verified directly by substitution. Thus, $\dot{\mathcal{V}}^{-1}(0) = \mathcal{U}_1 \cup \mathcal{U}_2$. 

Suppose that $\mathbf{X} : \mathbb{R}_{\geq 0} \longrightarrow \mathcal{U}_1 \cup \mathcal{U}_2$ is a solution of the NDBWSHCCM starting from the initial condition $\mathbf{X}(0) = \mathbf{Y} \in \mathbb{R}^3$. Suppose that there exists $t \in \mathbb{R}_{\geq 0}$ such that $\mathbf{X}(t) \in \mathcal{U}_1 \cap \mathcal{U}_2 = \{ 0 \}$. Since $0 \in \mathcal{E}$, by the uniqueness of the solutions, $\mathbf{X}(t) = 0$ for all $t \in \mathbb{R}_{\geq 0}$. Thence, $\mathcal{O}_{\mathbf{Y}}^{+} = \{ 0 \}$. Note that $\mathcal{U}_1 \setminus \{ 0 \}$ and $\mathcal{U}_2 \setminus \{ 0 \}$ are separated. Thus, since $\mathcal{O}_{\mathbf{Y}}^{+}$ is connected and $\mathcal{O}_{\mathbf{Y}}^{+} \subseteq \mathcal{U}_1 \cup \mathcal{U}_2$, by Lemma \ref{thm:con_sep}, either $\mathcal{O}_{\mathbf{Y}}^{+} = \{ 0 \}$ or $\mathcal{O}_{\mathbf{Y}}^{+} \subseteq \mathcal{U}_1 \setminus \{ 0 \}$ or $\mathcal{O}_{\mathbf{Y}}^{+} \subseteq \mathcal{U}_2 \setminus \{ 0 \}$:
\begin{compactitem}
\item Case I: $\mathcal{O}_{\mathbf{Y}}^{+} = \{ 0 \}$. Then, $\mathcal{O}_{\mathbf{Y}}^{+} \subseteq \mathcal{E}$.
\item Case II: $\mathcal{O}_{\mathbf{Y}}^{+} \subseteq \mathcal{U}_1 \setminus \{ 0 \}$. Note that $X(t) = 0$ for all $t \in \mathbb{R}_{\geq 0}$. Therefore, $\dot{X}(t) = 0$ for all $t \in \mathbb{R}_{\geq 0}$. Fix $t \in \mathbb{R}_{\geq 0}$. Note that $\dot{X}(t) = V(t) \neq 0$. Therefore, a contradiction is reached and $\mathcal{O}_{\mathbf{Y}}^{+} \not\subseteq \mathcal{U}_1 \setminus \{ 0 \}$.
\item Case III: $\mathcal{O}_{\mathbf{Y}}^{+} \subseteq \mathcal{U}_2 \setminus \{ 0 \}$. Thus, $V(t) = 0$ for every $t \in \mathbb{R}_{\geq 0}$. Therefore, $\dot{V}(t) = 0$ for all $t \in \mathbb{R}_{\geq 0}$. Fix $t \in \mathbb{R}_{\geq 0}$. Then, $\dot{V}(t) = - \kappa \abs{X(t)}^{p - 1} X(t) - \kappa_c \abs{Z(t)}^{p - 1} Z(t) = 0$. Thus, $(X(t), Z(t), V(t)) \in \mathcal{E}$ (see the proof of Proposition \ref{thm:BWSHCCM_eq}). Therefore, by the uniqueness of the solutions, $(X(t), Z(t), V(t)) \in \mathcal{E}$ for all $t \in \mathbb{R}_{\geq 0}$. Thus, $\mathcal{O}_{\mathbf{Y}}^{+} \subseteq \mathcal{E}$.
\end{compactitem}
In summary, $\mathcal{O}_{\mathbf{Y}}^{+} \subseteq \mathcal{E}$. Thus, the largest invariant set that is contained in $\mathcal{U}_1 \cup \mathcal{U}_2$ is also contained in $\mathcal{E}$. Since $\mathcal{E}$ is an invariant set, it is also the largest invariant set that is contained in $\mathcal{U}_1 \cup \mathcal{U}_2 = \dot{\mathcal{V}}^{-1}(0)$.
\end{proof}

Define $\mathcal{S}_{\mathbf{X}} \triangleq \{ \mathbf{Y} \in \mathbb{R}^3 : \mathcal{V} (\mathbf{Y}) \leq \mathcal{V} (\mathbf{X}) \}$ for all $\mathbf{X} \in \mathbb{R}^3$. Then,
\begin{prop}\label{thm:BWSHCCM_convergence}
Suppose that the restrictions on the values of the parameters stated above are amended as follows: $B \in \mathbb{R}_{> 0}$, $\mathit{\Gamma} \in (-B, B)$, $\sigma \in \mathbb{R}_{> 0}$. Furthermore, suppose that $\mathbf{X} : \mathbb{R}_{\geq 0} \longrightarrow \mathbb{R}^3$ is a solution of the NDBWSHCCM. Suppose also that $\mathbf{X}(0) = \mathbf{Y}$. Then, $\mathcal{O}_{\mathbf{Y}}^{+\infty}$ is a nonempty, compact, connected, invariant set such that $\lim_{t \rightarrow +\infty} \mathbf{X}(t) = \mathcal{O}_{\mathbf{Y}}^{+\infty}$ and $\mathcal{O}_{\mathbf{Y}}^{+\infty} \subseteq \mathcal{S}_{\mathbf{Y}} \cap \mathcal{E}$.
\end{prop}
\begin{proof}
Note that $\mathcal{O}_{\mathbf{Y}}^{+}$ is bounded by Proposition \ref{thm:BWSHCCM_EUB}. Thus, $\mathcal{O}_{\mathbf{Y}}^{+\infty}$ is a nonempty, compact, connected, invariant set, and $\lim_{t \rightarrow +\infty} \mathbf{X}(t) = \mathcal{O}_{\mathbf{Y}}^{+\infty}$ (e.g, see Proposition 5.1 in \cite{bhat_nontangency-based_2003} or Theorem 2.41 in \cite{haddad_nonlinear_2011}). Note that, by Lemma \ref{thm:BWSHCCM_inv_dV_1}, the largest invariant set contained in $\dot{\mathcal{V}}^{-1}(0)$ is $\mathcal{E}$. Suppose that $\mathcal{P}$ is the largest invariant set contained in $\mathcal{S}_{\mathbf{Y}}$. Then, $\mathcal{O}_{\mathbf{Y}}^{+\infty} \subseteq \mathcal{P} \cap \mathcal{E} \subseteq \mathcal{S}_{\mathbf{Y}} \cap \mathcal{E}$ (e.g., see Proposition 5.3 in \cite{bhat_nontangency-based_2003}).
\end{proof}

\section{Analysis of the NDBWMCM}\label{sec:analysis_BWMCM}

Consider the NDBWMCM given by Eq. \eqref{eq:BWMCM_nd} and Eq. \eqref{eq:BWMCM_nd_output}. Suppose that the NDBWMCM is parameterized by $B \in \mathbb{R}_{> 0}$, $\mathit{\Gamma} \in (-B, B)$, $\kappa \in (0, 1)$, $\sigma \in \mathbb{R}_{> 0}$, $n \in \mathbb{R}_{\geq 1}$, $p \in \mathbb{R}_{\geq 2} \cup \{ 1 \}$. Suppose also that $\kappa_c = 1 - \kappa$. $f : \mathbb{R}^4 \longrightarrow \mathbb{R}^4$ shall denote the state function associated with the NDBWMCM given by 
\[
\begin{cases}
f_1(\mathbf{X}) \triangleq \kappa \sigma \abs{Y}^{p - 1} Y + \kappa_c \sigma \abs{Z}^{p - 1} Z \\
f_2(\mathbf{X}) \triangleq W \\
f_3(\mathbf{X}) \triangleq W - B \abs{Z}^{n - 1} Z \abs{W} - \mathit{\Gamma} \abs{Z}^n W \\
f_4(\mathbf{X}) \triangleq - \frac{1}{\sigma} f_1(\mathbf{X}) - \kappa p \sigma \abs{Y}^{p - 1} f_2(\mathbf{X}) - \kappa_c p \sigma \abs{Z}^{p - 1} f_3(\mathbf{X}) \\
\end{cases}
\]
for all $\mathbf{X} \triangleq (R, Y, Z, W) \in \mathbb{R}^4$. Let $g : \mathbb{R}^4 \longrightarrow \mathbb{R}^2$ denote the output function associated with the NDBWMCM given by 
\[
\begin{cases}
g_1(\mathbf{X}) \triangleq R + Y\\
g_2(\mathbf{X}) \triangleq \kappa \sigma \abs{Y}^{p - 1} Y + \kappa_c \sigma \abs{Z}^{p - 1} Z + W
\end{cases}
\]
for all $\mathbf{X} \in \mathbb{R}^4$. Unless stated otherwise, the restrictions on the initial conditions that are stated in the main body of the article will be relaxed to $R(0) = R_0$, $Y(0) = Y_0$, $Z(0) = Z_0$, $W(0) = W_0$ with $R_0, Y_0, Z_0, W_0 \in \mathbb{R}$.

For notational convenience, introduce the abbreviation $Q \triangleq g_2$. Define the Lyapunov candidate $\mathcal{V} : \mathbb{R}^4 \longrightarrow \mathbb{R}$ given by 
\[
\mathcal{V} (\mathbf{X}) \triangleq \left( R + \sigma Q(\mathbf{X}) \right)^2 + \frac{Q(\mathbf{X})^2}{2} + \frac{\kappa}{p + 1} \abs{Y}^{p + 1} + \frac{\kappa_c}{p + 1} \abs{Z}^{p + 1}
\]
for all $\mathbf{X} \in \mathbb{R}^4$. Note that $\mathcal{V}$ is continuously differentiable, radially unbounded, $\mathcal{V}(0) = 0$, and $\mathcal{V}(\mathbf{X}) > 0$ for all $\mathbf{X} \in \mathbb{R}^4 \setminus \{ 0 \}$.\footnote{It should be noted that the proof that the Lyapunov candidate for the NDBWMCM is radially unbounded may not appear to be entirely trivial, but it is still a routine exercise in real analysis.}

Define $h : \mathbb{R}^4 \longrightarrow \mathbb{R}$ as 
\[
h(\mathbf{X}) \triangleq \kappa \abs{Y}^{p - 1} Y + \kappa_c\abs{Z}^{p - 1} Z
\]
Note that 
\[
\dot{\mathcal{V}}(\mathbf{X}) = - \sigma h(\mathbf{X})^2 - \kappa_c \abs{Z}^{p + n - 1} \left( B \abs{Z} \abs{W} + \mathit{\Gamma} Z W \right)
\]
Then, 
\begin{lem}\label{thm:BWMCM_dV_leq_0}
Under the restrictions on the values of the parameters stated above, $\dot{\mathcal{V}}(\mathbf{X}) \leq 0$ for all $\mathbf{X} \in \mathbb{R}^4$. 
\end{lem}
\begin{proof}
Note that $- \sigma h(\mathbf{X})^2 \leq 0$ for all $\mathbf{X} \in \mathbb{R}^4$. The remainder of the proof is essentially identical to a part of the proof of Lemma \ref{thm:BWSHCCM_dV_leq_0}.
\end{proof}

\begin{prop}\label{thm:BWMCM_EUB}
Under the restrictions on the values of the parameters stated above, there exists a unique solution of the NDBWMCM on any time interval $[0, T)$ with $T \in \mathbb{R}_{>0} \cup \{ +\infty \}$ for every initial condition $(R_0, Y_0, Z_0, W_0) \in \mathbb{R}^4$. Furthermore, the solutions and the outputs of the NDBWMCM are equi-bounded.
\end{prop}
\begin{proof}
The proof of the global existence, uniqueness, and equi-boundedness of the solutions follow from arguments similar to the ones that were used in the proof of Proposition \ref{thm:BWSHCCM_EUB}, taking into account Lemma \ref{thm:BWMCM_dV_leq_0}. The equi-boundedness of the outputs follows from Lemma \ref{thm:output_bound}.
\end{proof}

Define 
\[
\mathcal{E} \triangleq \left\{ \left( R, Y, - \left( \frac{\kappa}{\kappa_c}\right)^{\frac{1}{p}} Y, 0 \right) : R, Y  \in \mathbb{R} \right\}
\]
\begin{prop}\label{thm:BWMCM_eq}
Under the restrictions on the values of the parameters stated above, $\mathcal{E}$ is the set of all equilibrium points of the NDBWMCM.
\end{prop}
\begin{proof}
Suppose that $(R, Y, Z, W) \in \mathbb{R}^4$ is an equilibrium point of the NDBWMCM. Then, $W = 0$ and $\abs{Z}^{p - 1} Z = - (\kappa/\kappa_c) \abs{Y}^{p - 1} Y$. Therefore, taking into account the restrictions on the values of the parameters, it can be shown that $Z = -(\kappa/\kappa_c)^{1/p} Y$. Thus, $(R, Y, Z, W) \in \mathcal{E}$. Suppose that $(R, Y, Z, W) \in \mathcal{E}$. Then, $W = 0$ and $Z = - \left(\kappa/\kappa_c \right)^{1/p} Y$. Then, it can be verified by substitution that $f \left(R, Y, Z, W \right) = 0$. Thus, $(R, Y, Z, W)$ is an equilibrium point of the NDBWMCM.
\end{proof}

Define the following sets:
\[
\mathcal{U}_1 \triangleq \{ (R, 0, 0, W) : R, W \in \mathbb{R} \}
\] 
\[
\mathcal{U}_2 \triangleq \left\{ \left( R, Y, -\left( \frac{\kappa}{\kappa_c} \right)^{\frac{1}{p}} Y, 0 \right) : R, Y \in \mathbb{R} \right\} = \mathcal{E}
\] 
\begin{lem}\label{thm:BWMCM_inv_dV_1}
Under the restrictions on the values of the parameters stated above, $\dot{\mathcal{V}}^{-1}(0) = \mathcal{U}_1 \cup \mathcal{U}_2$ and the largest invariant set that is contained in $\dot{\mathcal{V}}^{-1}(0)$ is $\mathcal{E}$.
\end{lem}
\begin{proof}
Suppose that $\dot{\mathcal{V}}(R, Y, Z, W) = 0$. From the proof of Lemma \ref{thm:BWMCM_dV_leq_0},
\[
\begin{cases}
\kappa \sigma \abs{Y}^{p - 1} Y + \kappa_c \sigma \abs{Z}^{p - 1} Z = 0 \\
\kappa_c \abs{Z}^{p + n - 1} \left( B \abs{Z} \abs{W} + \mathit{\Gamma} Z W \right) = 0
\end{cases}
\]
Then, by the proof of Proposition \ref{thm:BWMCM_eq}, $Z = -\left( \kappa / \kappa_c \right)^{1/p} Y$. Thus, either $W = 0$ and $Z = -\left( \kappa /\kappa_c \right)^{1/p} Y$ or $Y = Z = 0$. That $\dot{\mathcal{V}}(R, Y, Z, W) = 0$ if either $W = 0$ and $Z = -\left( \kappa /\kappa_c \right)^{1/p} Y$ or $Y = Z = 0$ can be verified directly by substitution. Thence, $\dot{\mathcal{V}}^{-1}(0) = \mathcal{U}_1 \cup \mathcal{U}_2$.

Suppose that $\mathbf{X} : \mathbb{R}_{\geq 0} \longrightarrow \mathcal{U}_1 \cup \mathcal{U}_2$ is a solution of the NDBWMCM starting from the initial condition $\mathbf{X}(0) = \mathbf{Y} \in \mathbb{R}^4$. Since every point in $\mathcal{E}$ is an equilibrium point, by uniqueness of solutions, either $\mathbf{X} : \mathbb{R}_{\geq 0} \longrightarrow \mathcal{E}$ or $\mathbf{X} : \mathbb{R}_{\geq 0} \longrightarrow \mathcal{U}_1 \setminus \mathcal{E}$. Suppose that $\mathbf{X} : \mathbb{R}_{\geq 0} \longrightarrow \mathcal{U}_1 \setminus \mathcal{E}$. Note that $Y(t) = 0$ for all $t \in \mathbb{R}_{\geq 0}$. Therefore, $\dot{Y}(t) = 0$ for all $t \in \mathbb{R}_{\geq 0}$. Fix $t \in \mathbb{R}_{\geq 0}$. Note that $\dot{Y}(t) = W(t) \neq 0$. Therefore, a contradiction is reached. Thus, $\mathbf{X} : \mathbb{R}_{\geq 0} \longrightarrow \mathcal{E}$. Since $\mathcal{E}$ is an invariant set, it is also the largest invariant set that is contained in $\mathcal{U}_1 \cup \mathcal{U}_2 = \dot{\mathcal{V}}^{-1}(0)$.
\end{proof}

Define $\mathcal{S}_{\mathbf{X}} \triangleq \{ \mathbf{Y} \in \mathbb{R}^4 : \mathcal{V} (\mathbf{Y}) \leq \mathcal{V} (\mathbf{X}) \}$
for all $\mathbf{X} \in \mathbb{R}^4$. Then,
\begin{prop}\label{thm:BWMCM_convergence}
Suppose that the restrictions on the values of the parameters stated above are in effect. Suppose also that $\mathbf{X} : \mathbb{R}_{\geq 0} \longrightarrow \mathbb{R}^4$ is a solution of the NDBWMCM. Suppose also that $\mathbf{X}(0) = \mathbf{Y}$. Then, $\mathcal{O}_{\mathbf{Y}}^{+\infty}$ is a nonempty, compact, connected, invariant set such that $\lim_{t \rightarrow +\infty} \mathbf{X}(t) = \mathcal{O}_{\mathbf{Y}}^{+\infty}$ and $\mathcal{O}_{\mathbf{Y}}^{+\infty} \subseteq \mathcal{S}_{\mathbf{Y}} \cap \mathcal{E}$.
\end{prop}
\begin{proof}
Note that $\mathcal{O}_{\mathbf{Y}}^{+}$ is bounded by Proposition \ref{thm:BWMCM_EUB}. Thus, $\mathcal{O}_{\mathbf{Y}}^{+\infty}$ is a nonempty, compact, connected, invariant set, and $\lim_{t \rightarrow +\infty} \mathbf{X}(t) = \mathcal{O}_{\mathbf{Y}}^{+\infty}$ (e.g, see Proposition 5.1 in \cite{bhat_nontangency-based_2003} or Theorem 2.41 in \cite{haddad_nonlinear_2011}). Note that, by Lemma \ref{thm:BWMCM_inv_dV_1}, the largest invariant set contained in $\dot{\mathcal{V}}^{-1}(0)$ is $\mathcal{E}$. Suppose that $\mathcal{P}$ is the largest invariant set contained in $\mathcal{S}_{\mathbf{Y}}$. Then, $\mathcal{O}_{\mathbf{Y}}^{+\infty} \subseteq \mathcal{P} \cap \mathcal{E} \subseteq \mathcal{S}_{\mathbf{Y}} \cap \mathcal{E}$ (e.g., see Proposition 5.3 in \cite{bhat_nontangency-based_2003}).
\end{proof}

\bibliographystyle{asmejour} 

\bibliography{template.bib} 

\end{document}